\documentclass[11pt]{article}

\usepackage{titlesec}

\usepackage{amsfonts,amsmath,amssymb,amsthm,color}
\usepackage[margin=1.0in]{geometry}
\usepackage{microtype}
\definecolor{weborange}{rgb}{.8,.3,.3}
\definecolor{webblue}{rgb}{0,0,.8}
\definecolor{internallinkcolor}{rgb}{0,.5,0}
\definecolor{externallinkcolor}{rgb}{0,0,.5}

\providecommand{\remove}[1]{}

\usepackage{comment}

\usepackage{enumitem}
\usepackage[labelfont=bf]{caption}

\usepackage[
hyperfootnotes=false,
colorlinks=true,
urlcolor=externallinkcolor,
linkcolor=internallinkcolor,
filecolor=externallinkcolor,
citecolor=internallinkcolor,
breaklinks=true,
pdfstartview=FitH,
pdfpagelayout=OneColumn]{hyperref}

\usepackage[backend=bibtex,style=ieee-alphabetic,natbib=true,backref=true]{biblatex} 
\usepackage{cleveref}
\usepackage{xspace}
\usepackage{dsfont}

\usepackage{colortbl}
\usepackage{xcolor}
\usepackage{graphicx}

\providecommand{\remove}[1]{}

\ifdefined\IsDraft
\newcommand{\authnote}[2]{{\bf [{\color{red} #1's Note:} {\color{blue} #2}]}}
\else
\newcommand{\authnote}[2]{}
\fi

\remove{
	
	\titleclass{\subsubsubsection}{straight}[\subsection]
	
	\newcounter{subsubsubsection}[subsubsection]
	\renewcommand\thesubsubsubsection{\thesubsubsection.\arabic{subsubsubsection}}
	
	\titleformat{\subsubsubsection}
	{\normalfont\normalsize\bfseries}{\thesubsubsubsection}{1em}{}
	\titlespacing*{\subsubsubsection}
	{0pt}{3.25ex plus 1ex minus .2ex}{1.5ex plus .2ex}
	
	\makeatletter
	\renewcommand\paragraph{\@startsection{paragraph}{5}{\z@}%
		{3.25ex \@plus1ex \@minus.2ex}%
		{-1em}%
		{\normalfont\normalsize\bfseries}}
	\renewcommand\subparagraph{\@startsection{subparagraph}{6}{\parindent}%
		{3.25ex \@plus1ex \@minus .2ex}%
		{-1em}%
		{\normalfont\normalsize\bfseries}}
	\def\toclevel@subsubsubsection{4}
	\def\toclevel@paragraph{5}
	\def\toclevel@paragraph{6}
	\def\l@subsubsubsection{\@dottedtocline{4}{7em}{4em}}
	\def\l@paragraph{\@dottedtocline{5}{10em}{5em}}
	\def\l@subparagraph{\@dottedtocline{6}{14em}{6em}}
	\makeatother
	
	\setcounter{secnumdepth}{4}
	\setcounter{tocdepth}{4}
} 

\newenvironment{algorithm}{\begin{mybox} \vspace{-.1in}\begin{algo}}{ \vspace{-.1in} \end{algo}\end{mybox}}

\newenvironment{mybox}{\begin{center}\begin{tabular}{|p{0.97\linewidth}|c|}   \hline} {  \\ \hline \end{tabular} \end{center}}			

\let\originalleft\left
\let\originalright\right
\renewcommand{\left}{\mathopen{}\mathclose\bgroup\originalleft}
\renewcommand{\right}{\aftergroup\egroup\originalright}

	\newcommand{\wlg} {without loss of generality\xspace}

	\newcommand{\ip}[1]{\iprod{#1}}
	\newcommand{\iprod}[1]{\langle #1 \rangle}
	
	\newcommand{\set}[1]{\ens{#1}}

	\newcommand{\paren}[1]{\left(#1\right)}

	\newcommand{\eqdef}{:=}

	\newcommand{\N}{{\mathbb{N}}}
        \newcommand{\E}{{\mathbf{E}}}

	\newcommand{\zo}{\set{0,1}}
	\newcommand{\oo}{\set{-1,1}}

	\newcommand{\condition}{\;\ifnum\currentgrouptype=16 \middle\fi|\;}

	\newcommand{\xor}{\oplus}

	\newcommand{\eps}{\varepsilon}

	\newcommand{\la}{\gets}

	\newcommand{\Enc}{\mathsf{Enc}}
	\newcommand{\Dec}{\mathsf{Dec}}
	\newcommand{\Gen}{\mathsf{Gen}}

	\newcommand{\Attacker}{\mathsf{Attacker}}

	\newcommand{\negl}{\operatorname{neg}}

	\newcommand{\Supp}{\operatorname{Supp}}

	\newcommand{\Ensuremath}[1]{\ensuremath{#1}\xspace}

	\newcommand{\tth}[1]{\Ensuremath{#1^{\rm th}}}
	
	\newcommand{\ith}{\tth{i}}

	\renewcommand{\cref}{\Cref}
	\usepackage{aliascnt}
	
	\newtheorem{theorem}{Theorem}[section]
	
	\newaliascnt{lemma}{theorem}
	\newtheorem{lemma}[lemma]{Lemma}
	\aliascntresetthe{lemma}
	\crefname{lemma}{Lemma}{Lemmas}

	\newaliascnt{observation}{theorem}
	
	\aliascntresetthe{observation}
	\crefname{observation}{Observation}{Observation}

	\newaliascnt{claim}{theorem}
	\newtheorem{claim}[claim]{Claim}
	\aliascntresetthe{claim}
	\crefname{claim}{Claim}{Claims}
	
	\newaliascnt{corollary}{theorem}
	
	\aliascntresetthe{corollary}
	\crefname{corollary}{Corollary}{Corollaries}

	\newaliascnt{construction}{theorem}
	
	\aliascntresetthe{construction}
	\crefname{construction}{Construction}{Constructions}
	
	\newaliascnt{fact}{theorem}
	
	\aliascntresetthe{fact}
	\crefname{fact}{Fact}{Facts}
	
	\newaliascnt{proposition}{theorem}
	
	\aliascntresetthe{proposition}
	\crefname{proposition}{Proposition}{Propositions}
	
	\newaliascnt{conjecture}{theorem}
	
	\aliascntresetthe{conjecture}
	\crefname{conjecture}{Conjecture}{Conjectures}

	\newaliascnt{definition}{theorem}
	\newtheorem{definition}[definition]{Definition}
	\aliascntresetthe{definition}
	\crefname{definition}{Definition}{Definitions}
	
	\newaliascnt{notation}{theorem}
	
	\aliascntresetthe{notation}
	\crefname{notation}{Notation}{Notation}
	
	\newaliascnt{assertion}{theorem}
	
	\aliascntresetthe{assertion}
	\crefname{assertion}{Assertion}{Assertion}
	
	\newaliascnt{assumption}{theorem}
	
	\aliascntresetthe{assumption}
	\crefname{assumption}{Assumption}{Assumption}

	\newaliascnt{remark}{theorem}
	\newtheorem{remark}[remark]{Remark}
	\aliascntresetthe{remark}
	\crefname{remark}{Remark}{Remarks}
	
	\newaliascnt{question}{theorem}
	\newtheorem{question}[question]{Question}
	\aliascntresetthe{question}
	\crefname{question}{Question}{Question}

	\newaliascnt{example}{theorem}
	\ifdefined\excludeexample
	\else
	\newtheorem{example}[example]{Example}
	\fi

	\aliascntresetthe{example}
	\crefname{exmaple}{Example}{Examples}

	\crefname{equation}{Equation}{Equations}

	\newaliascnt{proto}{theorem}
	
	\newtheorem{proto}[proto]{Protocol}
	
	\aliascntresetthe{proto}
	\crefname{proto}{protocol}{protocols}

	\newaliascnt{algo}{theorem}
	\newtheorem{algo}[algo]{Algorithm}
	\aliascntresetthe{algo}
	\crefname{algo}{algorithm}{algorithms}

	\newaliascnt{func}{theorem}
	\newtheorem{func}[func]{Function}
	\aliascntresetthe{func}
	\crefname{func}{function}{functions}

	\newaliascnt{expr}{theorem}
	\newtheorem{expr}[expr]{Experiment}
	\aliascntresetthe{expr}
	\crefname{experiment}{experiment}{experiments}

	\newaliascnt{gm}{theorem}
	\newtheorem{gm}[gm]{Game}
	\aliascntresetthe{gm}
	\crefname{game}{game}{games}

	\newcommand{\stepref}[1]{Step~\ref{#1}}

	\def\FullBox{$\Box$}
	\def\qed{\ifmmode\qquad\FullBox\else{\unskip\nobreak\hfil
			\penalty50\hskip1em\null\nobreak\hfil\FullBox
			\parfillskip=0pt\finalhyphendemerits=0\endgraf}\fi}
	
	\def\qedsketch{\ifmmode\Box\else{\unskip\nobreak\hfil
			\penalty50\hskip1em\null\nobreak\hfil$\Box$
			\parfillskip=0pt\finalhyphendemerits=0\endgraf}\fi}

	{\proof[#1]\list{}{\leftmargin=1cm\rightmargin=1cm}}
	{\endproof%
		\vspace{10pt}
	}
	
	\newcommand{\ex}[1]{\Ex\left[#1\right]}
	\newcommand{\Ex}{{\mathrm E}}
	\renewcommand{\Pr}{{\mathrm {Pr}}}
	\newcommand{\pr}[1]{\Pr\left[#1\right]}
	\newcommand{\ppr}[2]{\Pr_{#1}\left[#2\right]}

	\newcommand{\NHdist}{\mathsf{NHamDist}}
	
	\newcommand{\Ac}{\mathsf{A}}

	\newcommand{\cC}{{\mathcal{C}}}

	\newcommand{\ens}[1]{\{#1\}}
	\newcommand{\size}[1]{\left|#1\right|}

	\newcommand{\cM}{{\cal{M}}}
	\providecommand{\cL}{{\cal{L}}}

	\newcommand{\ppt}{{\sc ppt}\xspace}
	
	\def\cA{{\cal A}}
	\def\cB{{\cal B}}
	\def\cC{{\cal C}}
	\def\cD{{\cal D}}
	
	\def\cF{{\cal F}}

	\def\cL{{\cal L}}
	\def\cM{{\cal M}}

	\def\cP{{\cal P}}

	\def\cS{{\cal S}}
	
	\def\cU{{\cal U}}
	\def\cV{{\cal V}}
	
	\def\cX{{\cal X}}
	
	\def\cY{{\cal Y}}
	\def\cZ{{\cal Z}}

	\def\bbN{{\mathbb N}}
	
	\def\bbR{{\mathbb R}}

	\newcommand{\indic}[1]{\mathds{1}{\set{#1}}}

	\newcommand{\pt}[1]{\boldsymbol{#1}}

	\newcommand{\y} {\pt{y}}
	
	\newcommand{\z} {\pt{z}}

	\newcommand{\tp} {\tilde{p}}

	\newcommand{\by}{\bold{y}}
	
	\newcommand{\bb}{\bold{b}}
	\newcommand{\bp}{\bold{p}}

	\newcommand{\bi}{\bold{i}}
	
	\newcommand{\bx}{\bold{x}}

	\newcommand{\bS}{\bold{S}}
	\newcommand{\bT}{\bold{T}}

	\newcommand{\bz}{\bold{z}}

	\newcommand{\poly}{{\rm poly}}

	\makeatletter
	\let\xx@thm\@thm
	\AtBeginDocument{\let\@thm\xx@thm}
	\makeatother

        \renewcommand{\epsilon}{\varepsilon}

\newcommand{\Enote}[1]{\authnote{Eliad}{#1}}
\newcommand{\Unote}[1]{\authnote{Uri}{#1}}

\newcommand{\haim}[1]{\authnote{Haim}{#1}}

\newcommand{\MMM}{{\cal M}}
\newcommand{\FFF}{{\cal F}}
\newcommand{\DDD}{{\cal D}}
\newcommand{\XXX}{{\cal X}}
\newcommand{\AAA}{{\cal A}}
\newcommand{\BBB}{{\cal B}}
\newcommand{\CCC}{{\cal C}}
\newcommand{\UUU}{{\cal U}}

\newcommand{\ignore}[1]{}

\title{Data Reconstruction: When You See It and When You Don't} 

\ifdefined\IsAnonymous
\author{Anonymized for Submission}
\else
\author{
Edith Cohen\thanks{Google Research and Tel Aviv University. {\tt edith@cohenwang.com}} 
\and
Haim Kaplan\thanks{Tel Aviv University and Google Research. {\tt haimk@tau.ac.il}.}
\and
Yishay Mansour\thanks{Tel Aviv University and Google Research. {\tt mansour.yishay@gmail.com}.}
\and
Shay Moran\thanks{Departments of Mathematics, Computer Science, and Data and Decision Sciences, Technion and Google Research. {\tt smoran@technion.ac.il}}
\and
Kobbi Nissim\thanks{Department of Computer Science, Georgetown University. {\tt kobbi.nissim@georgetown.edu}. Work done while at Google Research, Tel-Aviv.}
\and
Uri Stemmer\thanks{Tel Aviv University and Google Research. {\tt u@uri.co.il}.}
\and
Eliad Tsfadia\thanks{Department of Computer Science, Georgetown University. \texttt{eliadtsfadia@gmail.com}.}
}
\fi
\begin{document}
	
	\maketitle
	
	\begin{abstract}
We revisit the fundamental question of formally defining what constitutes a {\em reconstruction attack}. While often clear from the context, our exploration reveals that a precise definition is much more nuanced than it appears, to the extent that a single all-encompassing definition may not exist. Thus, we employ a different strategy and aim to ``sandwich'' the concept of reconstruction attacks by addressing two complementing questions: (i) What conditions guarantee that a given system is protected against such attacks? (ii) Under what circumstances does a given attack clearly indicate that a system is not protected? More specifically,
\begin{itemize}
    \item We introduce a new definitional paradigm -- {\em Narcissus Resiliency} -- to formulate a security definition for protection against reconstruction attacks. This paradigm has a self-referential nature that enables it to circumvent shortcomings of previously studied notions of security.

    Furthermore, as a side-effect, we demonstrate that Narcissus resiliency captures as special cases multiple well-studied concepts including differential privacy and other security notions of one-way functions and encryption schemes.
    
    \item We formulate a link between reconstruction attacks and {\em Kolmogorov complexity}. This allows us to put forward a criterion for evaluating when such attacks are convincingly successful.
\end{itemize}

		
	\end{abstract}
	

\section{Introduction}

Reasoning about data privacy is crucial in today's data-driven world. This includes the design of privacy-enhancing technologies aimed at protecting privacy, as well as identifying vulnerabilities of existing methods by conducting ``privacy attacks''. 
The most severe family of attacks is, arguably, {\em reconstruction attacks}, where an adversary takes what appears to be benign statistics and reconstructs significant portions of the sensitive data~\citep{DiNi03}.
Such attacks have helped shape the theory of data privacy as well as expose vulnerabilities in existing real-world systems. 

However, prior works did not coalesce around a single {\em definition of reconstruction} and instead considered several context-dependent definitions. Although these definitions made sense in the context in which they were introduced, they do not necessarily carry over to other settings, as we will later explain. Motivated by this, our work is set to answer the following meta-question:

\begin{question}\label{q:motivation}
What is reconstruction?    
\end{question}

This question is more nuanced than it initially appears, as will become clear later. It seems that for every attempt to define reconstruction mathematically, there are cases that either fit the definition but do not ``feel like'' reconstruction, or vice versa. 
As a result, we do not know whether a precise answer to Question~\ref{q:motivation} exists, and leave it as an open question for future work. 
In this work, we aim to approach Question~\ref{q:motivation} by ``sandwiching'' the concept of reconstruction and studying the following two questions:

\begin{question}\label{q:protect}
What conditions are sufficient to ensure that a given system is protected against reconstruction attacks?
\end{question}

\begin{question}\label{q:attack}
Under what circumstances does an attack clearly indicate that a given system is vulnerable?
\end{question}

This is reminiscent of the current state of affairs in the related concept of {\em differential privacy (DP)} \citep{DMNS06}: If an algorithm satisfies DP (with small enough privacy parameters), then it is guaranteed to be ``safe'' in terms of its privacy implications. However, the fact that an algorithm does not satisfy DP does not immediately mean that it is unsafe privacy-wise. 
Rather, to convince that an algorithm is unsafe, one usually needs to demonstrate an actual attack, such as a reconstruction or a membership attack. In other words, in the context of data privacy, we currently do not have a good definition of what constitutes ``privacy''; only definitions of what it means to ``protect'' privacy, and definitions for what it means to ``attack'' an algorithm in order to show a privacy breach. 
We embrace this way of thinking and leverage it towards formally reasoning about reconstruction, as stated in Questions~\ref{q:protect} and~\ref{q:attack}. We elaborate on each of these two questions separately, in 
Sections~\ref{sec:introProtect} and~\ref{sec:introIdentify}, respectively. 

\subsection{Protecting against reconstruction}\label{sec:introProtect}

Consider an algorithm $\MMM$ that takes a dataset $S$ and returns an output $y$. For example, the dataset~$S$ might contain images, and the output $y$ may be an image generative model.
This output~$y$ is then given to a ``privacy attacker'' $\AAA$ whose goal is to output a ``reconstruction'' $z$ of elements from $S$.  Informally, we want to say that algorithm $\MMM$ {\em prevents reconstruction attacks} if $y$ is ``hard to invert'' in the sense that $y$ does not help the attacker in ``recovering'' elements from $S$. However, there are two immediate issues that require attention here:
\begin{enumerate}

\item What it means to ``recover elements'' is context-dependent. For example, our criteria for when one image ``recovers'' another image would likely differ from our criteria for when a text file ``recovers'' another text file. As we aim for a general treatment, we abstract this criteria away in the form of a {\em relation $R$}, where $R(S,z)=1$ if and only if $z$ is a ``valid reconstruction'' of~$S$ (or of some data point in $S$). In our example with the images, this could mean that $z$ is an image which is ``close enough'' to one of the images in $S$. Alternatively, depending on the context, this could also mean that $z$ is a collection of 100 images of which one is ``close enough'' to an image in $S$, or it could mean that $z$ is a vector of 100 images, each of them is ``somewhat close'' to a distinct image in $S$, etc.

    \item The underlying distribution of the data matters a lot. For example, consider a scenario where every sufficiently large image dataset contains the Mona Lisa. An `attack' that outputs an image that is similar or identical to the Mona Lisa is not necessarily a successful attack, as this could be accomplished without even accessing the generative model $y$.
\end{enumerate}

These considerations led \citet{balle2022reconstructing} and \citet{cummings2024attaxonomy} to present 
definitions with the following flavor:

\begin{definition}[\cite{balle2022reconstructing,cummings2024attaxonomy}]\label{def:cummings}
Let $\XXX$ be a data domain, let $\DDD$ be a distribution over datasets containing elements from $\XXX$, and let $R:\XXX^*\times\{0,1\}^*\rightarrow\{0,1\}$ be a reconstruction criterion.
Algorithm $\MMM$ is $(\eps,\delta,\DDD)$-$R$-reconstruction-robust if for all attackers $\AAA$ it holds that
\begin{equation}\label{eq:cummings}
\underset{\substack{S\leftarrow\DDD\\y\leftarrow\MMM(S)\\z\leftarrow\AAA(y)}}{\Pr}[R(S,z)=1]\leq e^{\eps}\cdot \sup_{z^*}\underset{\substack{T\leftarrow\DDD}}{\Pr}[R(T,z^*)=1]+\delta.    
\end{equation}
\end{definition}

In words, in Definition~\ref{def:cummings} the attacker's probability of success is compared to a {\em baseline} which is the probability of success of the best {\em trivial attacker} that simply chooses a fixed element $z^*$ (independently of the dataset).
To contradict the security of algorithm $\MMM$, the attacker must succeed with a probability noticeably higher than the baseline (depending on the parameters $\eps$ and $\delta$).
Note that the aforementioned attacker that ``recovered'' the Mona Lisa would {\em not} contradict the security of $\MMM$. 
The reason is that if every (large enough) dataset sampled from $\DDD$ would contain the Mona Lisa, then in Definition~\ref{def:cummings} we could take $z^*$ to be the Mona Lisa, and hence the right hand size of Inequality~(\ref{eq:cummings}) would be one.\footnote{For this example we assume that if $z^*\in T$ then $R(T,z^*)=1$.} 

This introduces a serious problem: once the baseline probability is 1 then {\em no adversary} could ever contradict Inequality~(\ref{eq:cummings}), even adversaries that do achieve meaningful reconstructions. To illustrate the issue, suppose that the dataset contains both ``canonical'' photos (such as the Mona Lisa or canonical photos of US presidents), as well as ``sensitive'' photos of ordinary individuals. While recovering a canonical photo should not be considered a successful attack, recovering a ``sensitive'' photo definitely should. Formally, suppose that every image in $S$ is sampled independently as follows: with probability 1/2 return the Mona Lisa, and otherwise return a photo of a random citizen. 
If under these conditions $\AAA(\MMM(S))$ is able to recover a photo of an ordinary individual from $S$, then that should be considered a reconstruction attack. But this is not captured by Definition~\ref{def:cummings} as the performance of $\AAA$ is compared with a baseline where the Mona Lisa is recovered. More generally, a sever flaw of Definition~\ref{def:cummings} is that once the baseline probability is $\approx 1$ then no adversary could ever contradict Inequality~(\ref{eq:cummings}).

\subsubsection{Towards a new definition}

A takeaway from the above discussion is that comparing all adversaries to the same fixed baseline can create a problem. 
In particular, adversaries that identify ``special'' elements in $S$ (which are unlikely to appear in fresh datasets) should be distinguished from adversaries that identify ``trivial'' elements in $S$. 
We now make an attempt to incorporate this into the definition. As we will see, this attempt has different shortcomings.

\begin{definition}\label{def:weightedBaseline}
Let $\XXX$ be a data domain, let $\DDD$ be a distribution over datasets containing elements from $\XXX$, and let $R:\XXX^*\times\{0,1\}^*\rightarrow\{0,1\}$. 
Algorithm $\MMM$ is $(\eps,\tau,\DDD)$-$R$-reconstruction-robust if for all attackers $\AAA$ it holds that
\begin{equation}\label{eq:weightedBaseline}
\underset{\substack{S\leftarrow\DDD\\y\leftarrow\MMM(S)\\z\leftarrow\AAA(y)}}{\Pr}\Big[
R(S,z)=1
\;\text{ and }\;
\Pr_{T\leftarrow\DDD}[R(T,z)=1]\leq\tau
\Big]\leq \epsilon.    
\end{equation}
\end{definition}

In words, with this definition the attacker's goal is to identify an element $z$ such that (1) $z$ is a valid reconstruction w.r.t.\ the dataset $S$; and (2) the same $z$ is unlikely to be a valid reconstruction w.r.t.\ a fresh dataset $T$. 
Informally, $\Pr_{T\leftarrow\DDD}[R(T,z)=1]$ serves as a ``conditional baseline'' that adapts itself to the element $z$ chosen by the attacker. Note that in the context of our example with the Mona Lisa, an attacker that ``recovers'' the Mona Lisa would not contradict Inequality~(\ref{eq:weightedBaseline}), because 
$\Pr_{T\leftarrow\DDD}[R(T,z)=1]$ would be large and hence greater than the threshold $\tau$. 
On the other hand, an attacker that manages to recover the photo of an ordinary individual from $S$ (with large enough probability) would contradict Inequality~(\ref{eq:weightedBaseline}). So Definition~\ref{def:weightedBaseline} achieves the desired behavior for this example.

A shortcoming of Definition~\ref{def:weightedBaseline} is that the values of $\eps$ and $\tau$ are necessarily context dependent. To illustrate this, we will now describe two situations where in one of them we need to set $\eps,\tau<\frac{1}{2}$ in order for the definition to make sense, while in the other situation $\eps$ and $\tau$ must be larger than $\frac{1}{2}$.

\begin{itemize}
   \item For the first situation, consider the following distribution $\DDD$ over datasets: With probability $\frac{1}{2}$ return a dataset containing random photos of ordinary citizens. Otherwise, return a dataset sampled similarly, except that one of its photos is replaced with the Mona Lisa. Under these conditions, we must set $\tau<\frac{1}{2}$ in order to circumvent the ``trivial attack'' using the Mona Lisa. Setting $\eps$ close to 0 seems reasonable under these conditions in order to guarantee that a meaningful reconstruction can happen only with small probability. 
    
   \item For the second situation, consider a case where the dataset is a vector of $n$ random bits, and the attacker's goal is to pinpoint a single entry from this vector and to guess it with high probability. Formally, in this example we interpret the outcome of the adversary as $z=(i,b)\in[n]\times\{0,1\}$, where $R(S,(i,b))=1$ if and only if $S[i]=b$. In this scenario, \emph{every} adversary always satisfies $\Pr_{T\leftarrow\DDD}[R(T,z)=1]=1/2$ and hence if we set $\tau<1/2$ then \underline{every} algorithm is safe w.r.t.\ to that definition, which is absurd. Therefore we may assume that $\tau\geq1/2$. In this case the condition $\Pr_{T\leftarrow\DDD}[R(T,z)=1]\leq \tau$ holds for every adversary and hence is redundant and we are only left with the first condition. How about~$\eps$? Note that a trivial adversary, which always guesses that the first bit is~$1$, succeeds with probability $1/2$. Hence, $\eps$ must be $\geq 1/2$ or else \underline{no} algorithm can be safe. To summarize, if $\tau<1/2$ then every algorithm is safe, and else, if $\tau\geq 1/2$ and $\eps <1/2$ then no algorithm is safe. Hence, this condition is non-trivial only in the range $\eps,\tau \geq 1/2$.
\end{itemize}

Intuitively, the issue about $\eps,\tau$ being context dependent stems from their somewhat nonstandard semantics (especially the parameter $\tau$). 
Ideally, we would want $\eps$ and $\tau$ to quantify some notion of ``distance from optimality'' (similarly to common definitions in the literature, such as differential privacy). But this is not the case with 
Definition~\ref{def:weightedBaseline}. Rather, the threshold $\tau$ dictates what we consider to be a ``non-trivial'' baseline. But ``non-trivial'' is context dependent, and hence so is $\tau$.


\subsubsection{A new Definitional paradigm: Narcissus resiliency -- an adversary trying to beat itself in its own game}

So far, we have established that: (1) We need to compare the adversary's success probability to a baseline; (2) This baseline cannot be fixed for all adversaries; and (3) We want to circumvent the need for a hyperparameter controlling what ``non-trivial'' is, as this is likely to be context dependent. 
To tackle this, we present a new definitional paradigm which we call {\em Narcissus resiliency}. In this paradigm, instead of explicitly setting and tuning the baseline, we {\em let the attacker be its own baseline}.\footnote{The name ``Narcissus'' is inspired by the Greek myth of Narcissus, who became obsessed with his own reflection. This alludes to how, in our proposed paradigm, the adversary competes against itself rather than external benchmarks.} More specifically, the attacker's success probability on the real dataset is compared with its success probability in a baseline setting where it does not get any information about the real dataset. Furthermore, the attacker must use the same strategy in both settings. This is enforced by preventing the attacker from receiving any information that can help it distinguish whether it is examined with respect to real dataset $S$ or a fresh dataset $T$ sampled from the same distribution. 
Hence, a Narcissus attacker needs to choose a compromise strategy: on one hand it needs to be strong (so it succeeds with high probability when it is applied to $S$) and at the same time it needs to be weak (so it succeeds with low probability when it is applied to $T$). The formal definition follows:\footnote{In the following definition we switch from a single data distribution $\DDD$ to a family of data distributions $\FFF$, where the mechanism is required to be ``secure'' w.r.t.\ every distribution in the family. This is standard, and can also be applied in the context of Definition~\ref{def:cummings}.}

\begin{definition}[Narcissus resiliency]\label{def:narc}
Let $\XXX$ be a data domain, let $\FFF$ be a family of distributions over datasets containing elements from $\XXX$, and let $R:\XXX^*\times\{0,1\}^*\rightarrow\{0,1\}$.
Algorithm $\MMM$ is $(\eps,\delta,\FFF)$-$R$-Narcissus-resilient if for all $\DDD\in\FFF$ and for all attackers $\AAA$ it holds that
\begin{equation}\label{eq:narc}
\underset{\substack{S\leftarrow\DDD\\y\leftarrow\MMM(S)\\z\leftarrow\AAA(y)}}{\Pr}[R(S,z)=1]\leq e^{\eps}\cdot \underset{\substack{S\leftarrow\DDD\\T\leftarrow\DDD\\y\leftarrow\MMM(S)\\z\leftarrow\AAA(y)}}{\Pr}[R(T,z)=1]+\delta.
\end{equation}
\end{definition}

Note that in both experiments the adversary $\AAA$ is executed on a dataset $S$ sampled exactly in the same way. Thus, the adversary behaves exactly the same way in both executions. Nevertheless, the adversary's goal is to ``separate'' the two experiments, where in the left experiment it aims to maximize the probability of reconstruction while in the right experiment it aims to minimize this probability.

\begin{example}
Note that Definition~\ref{def:narc} achieves the desired behavior in our example with the Mona Lisa: In order to contradict Inequality~(\ref{eq:narc}), the attacker must recover (with noticeable probability) the photo of a citizen from $S$.
\end{example}

\begin{example}
Let the underlying dataset distribution $\DDD$ be uniform over $n$-bit vectors. Let $k\in[n]$ be a parameter and suppose that, for the sake of this example, the relation $R$ is defined as follows. Given a dataset $S\in\{0,1\}^n$ and the outcome of the attacker $z$, parse $z$ as a vector of $k$ distinct indices $\tilde{I} = (i_1,\ldots,i_k)\in[n]^k$ and $k$ values $\tilde{Z}= (\tilde{z}_1,\ldots, \tilde{z}_k)\in\{0,1\}^k$ and let $R(S,z)=1$ if and only if $S|_{\tilde{I}}=\tilde{Z}$. That is, in this example, the adversary's goal is to pinpoint $k$ coordinates from $S$ and to guess them correctly. In order to contradict Inequality~(\ref{eq:narc}), the attacker must recover the values of $k$ coordinates from $S$ with probability noticeably higher than $2^{-k}$.
\end{example}

We arrived at Definition~\ref{def:narc} through the lens of data reconstruction. To provide further evidence supporting the paradigm of Narcissus resiliency, we show in Section~\ref{sec:NarcissusExpressiveness} that Narcissus resiliency captures concepts from cryptography and privacy as special cases. This includes differential privacy, resiliency to membership inference attacks, security of one-way functions, security of encryption schemes, and more. We show that all of these concepts can be stated in the terminology of Narcissus resiliency, where an adversary is ``trying to beat itself in its own game''. Furthermore, the paradigm enables us to express strictly weaker protections compared to differential privacy that are still meaningful. To support this claim, in \cref{sec:example} we provide an example of a very natural \emph{deterministic} mechanism (which is clearly not differentially private) that estimates a counting query under Narcissus-resiliency w.r.t. the family of all i.i.d.\ distributions and a natural reconstruction predicate $R$.

\subsection{Identifying reconstruction}\label{sec:introIdentify}

So far we have discussed which types of mechanisms prevent reconstruction of the input data. Such definitions serve as important guidance for responsible algorithm design. However, to show that a mechanism is vulnerable to reconstruction in the real world, an attacker usually only has access to the output of the mechanism, which in the learning context is a model that can be seen as an interactive program, or even just as a fixed string (e.g., a real-world attacker would like to attack a given chat-bot model, and not the algorithm that created the model based on training data). Moreover, such attacker might not even be aware of the training procedure exactly.

Mathematically, in Section~\ref{sec:introProtect} we considered a setting where initially a dataset $S$ is sampled, and then an outcome $y$ is computed by an algorithm $\MMM$ based on $S$.
Our definition of 
Narcissus resiliency \ref{def:narc}
 pinpoints our ultimate desire \underline{from the algorithm $\MMM$}.
Often, however, $\MMM$ is not explicit and we only have access to a fixed model $y$ and a fixed training set $S$. So there is no clear distribution over $S$ and it is not clear what is the process which was used to compute $y$ from $S$. Still, in order to be able to assess whether a particular attack is successful, we would like to have a definition of what it means to ``reconstruct'' a particular $S$ from a particular $y$. 
To emphasize this difference, when $S$ and $y$ are fixed and we want to classify a successful attack (rather than resiliency to such)
we use the term ``extract'' instead of ``reconstruct''.

\begin{question}\label{question:extraction:K}
    How do we define that a {\em fixed} dataset $S$ (or a {\em fixed} point in it $x\in S$) is extractable from a {\em fixed} string $y$?
\end{question}

Intuitively, such ``extraction'' means that $x$ is encoded in $y$. But capturing this intuition more formally turns out to be challenging. In this work, we offer the first formal definition that, we hope, will help understand what real-world attacks do, and what properties future attacks should highlight in order to argue about their quality.

Prior works presented several definitions with the following flavor:

\begin{definition}[\cite{CarliniLargeModels21,CarliniDiffusion23}]\label{def:intro:Carlini}
    Let $R$ be a relation such that $R(x,z)=1$ means that $z$ is a valid extraction of $x$.  
    A string $x$ is \emph{extractable} from a model $y$, if there exists an efficient program $\cA$ such that $R(x,\cA(y)) = 1$.
\end{definition}

Intuitively, the program $\cA$ in \cref{def:intro:Carlini} serves as evidence that $x$ is extractable from $y$, because given $y$ it outputs $z$ with $R(x,z)=1$. For example, if $x$ appears at the \ith location of $y$, then the program that given $y$ as input, outputs $y_{i}, y_{i+1}, \ldots, y_{i+\size{x}-1}$ satisfy the condition of \cref{def:intro:Carlini}.

However, at the formal level, this definition is meaningless: Because $\cA$ is chosen \emph{after} the example $x$, then formally, every example $x$ is ``extractable'' as there exists $\cA$ that (ignores its input) and outputs $x$. \citet{CarliniLargeModels21} mentioned that in order to prevent such pathological cases, the program $\cA$ should be \emph{shorter} than $x$ (and therefore $\cA$ cannot ``memorize'' $x$). But this by itself is not enough. For example, consider the string $x = \underbrace{34\:34 \:\ldots\:34}_{1000 \text{ times}}$ and the program $\cA$ ``Print `$34$' $1000$ times''. Namely, this example illustrates that $\cA$ being short is clearly not evidence for not being able to ``memorize'' $x$.
Therefore, the length of $\cA$ cannot be the only criteria for determining its validity, and there should be some connection between the length of the attack $\cA$, and the ``complexity'' of the target output. Intuitively, if there is no way to encode the target output $x$ using a short program (i.e., to compress it), then by demonstrating a short attacking program that is able to reveal $x$ given $y$, proves that $x$ is really extracted from the model $y$.

\begin{question}
	How do we quantify the ``complexity'' of an outcome $x$?
\end{question}

This question leads us to a classical complexity measure, called \emph{Kolmogorov-Complexity}.

\subsubsection{Defining Extraction via Kolmogorov Complexity}\label{sec:RecognizingRec:intro:OurApproach}

What makes the string 343434343434343 less random than 285628563123452?
The notion of \emph{Kolmogorov complexity} (in short, $K$-complexity), introduced by \citet{Sol67,Kolmogorov68,Chaitin69} in the 60s, provides an elegant method for measuring the amount of ``randomness'' in  an individual string. Informally, the $K$-complexity of a string $x$, denoted by $K(x)$, is the length of the shortest program that outputs the string $x$.
Intuitively, a string $x$ has high $K$-complexity, if there is no short program that outputs it and halts.

One issue with $K$-complexity is that it is not well defined without specifying which programming language $\cL$ do we use. In theoretical results (when constants do not matter), this complexity is usually defined w.r.t. a fixed \emph{Universal Turing Machine}. But when constants matter, a more formal way to define it is as follows.

\begin{definition}[$K_{\cL}$-Complexity]\label{def:intro:KLcomplexity}
Let $\cL$ be a programming language (e.g., Python).
The $K_{\cL}$-complexity of a string $x$, denoted by $K_{\cL}(x)$, is the length of the shortest $\cL$-program that outputs $x$ and halts.
Similarly, given a set of strings $X$, we denote by $K_{\cL}(X)$ the length of the shortest $\cL$-program that outputs an element in $X$ and halts. 

\end{definition}

Now that we have \cref{def:intro:KLcomplexity} in hand, we are finally ready to define extraction.

\begin{definition}[Our extraction definition, informal]\label{def:intro:extraction}
    Let $R$ be an extraction relation and $\cL$ a programming language.
    We say that a string $x$ is $(R,\cL)$-extractable from a string $y$ iff there exists an $\cL$-program $\cA$ such that the following holds:
    \begin{enumerate}
        \item $\cA(y)$ outputs $z$ such that $R(x,z) = 1$, and
        \item $K_{\cL}(\set{z \colon R(x,z)=1}) \gg \size{\cA}$.\label{item:intro:HardExtraction}
    \end{enumerate}
    We measure the \emph{quality} of the extraction by  $1 - \frac{\size{\cA}}{K_{\cL}(\set{z \colon R(x,z)=1})}$ (closer to $1$ means a more significant extraction).
\end{definition}

Namely, $x$ is ``$(R,\cL)$-extractable'' from $y$ if there exists a \emph{short} $\cL$-program $\cA$ that outputs an extraction (according to $R$) of $x$ from $y$, while there is no short $\cL$-program that extracts $x$ without $y$. We note that extraction must be in a context of a specific programming language, since otherwise, for every $x$ and $y$ we could always find a ``special'' programming language $\cL_{x,y}$ such that $x$ is $\cL_{x,y}$-extractable from $y$.\footnote{For example, consider a programming language $\cL_{x,y}$ that given a command $y$ outputs $x$.} Satisfying this definition w.r.t. a \emph{standard} programming language (e.g., Python, Java, etc) should indeed be considered as a valid extraction. 

\cref{def:intro:extraction} can easily be extended to datasets $S$ (rather than a single example $x$), and in \cref{sec:RecognizingRec} we also consider a probabilistic version of the $K_{\cL}$-complexity that allows some error probability, and redefine \cref{def:intro:extraction} w.r.t. this version. 

To provide further evidence supporting \cref{def:intro:extraction}, we demonstrate in \cref{sec:real-world-attacks} how three different types of real-world attacks \cite{CarliniLargeModels21,HaimVYSI22,CarliniDiffusion23} can be explained using the terminology of \cref{def:intro:extraction}.

\subsubsection{Narcissus-Resiliency prevents non-trivial extraction of training data}\label{sec:intro:NR-prevent-Ext}

In \cref{sec:NR-prevents-extraction} we show that if $\cM$ is $(\eps,\delta,\FFF)$-$R$-Narcissus-resilient then it prevents extraction in the following sense. Fix a distribution $\DDD\in\FFF$, sample two independent datasets $S$ and $T$ from $\DDD$, and compute $y\leftarrow\MMM(S)$. We treat $S$ as the ``real'' dataset and $T$ as a ``shadow'' dataset. Suppose that there is no adversary $\BBB$ that given $y$ can find a short program $\AAA$ such that $R(T,\AAA(y))=1$. Informally, this means that extracting information about the shadow dataset $T$ without receiving any information about it is hard. Then, if $\MMM$ is Narcissus-resilient, there is also no adversary that achieves this w.r.t.\ $S$ (even though the adversary gets $y$ which was computed based on $S$).

\subsubsection{Verifying the validity of reconstruction attacks}\label{sec:intro:ver_hardness}

The main limitation of \cref{def:intro:extraction} is that \cref{item:intro:HardExtraction} (lower-bounding the Kolmogorov Complexity) is not verifiable, as computing the $K$-complexity is an intractable problem. We could consider a tractable version of it, called \emph{time-bounded} $K$-complexity, where given a parameter $t$ we consider only programs that halt within $t$ steps (\cite{Kolmogorov68,Sisper83,Trakhtenbrot84,Ko86}). This, however, would only relax intractability to inefficiency, which does not help in our context where we would like to verify a reconstruction attack efficiently. In practice, we can only use Heuristics to gain some level of confidence. E.g., to apply many well-known compression algorithms on $x$ and check that all of them results with a compressed representation that is much longer than $\size{\cA}$. Yet, we remark that any such Heuristic can fail to determine some highly compressible patterns, and as we demonstrate in \cref{sec:hard-of-ver}, this problem is indeed inherent assuming that a basic cryptography primitive (pseudo-random generator) exists.

\subsection{Additional related works}\label{sec:relatedWorks}

\paragraph{Memorization}

Perhaps the most basic privacy violation is \emph{memorization}.
\citet{FeldmanZ20,Feldman20,BBFST21,BBS22,Livni23,ADHLR24} theoretically study the necessity of memorization in learning. Roughly speaking, they show that for some tasks, the output of any accurate algorithm must have large \emph{mutual information} with the training data. A similar result of \citet{BZ16} show that there exist learning tasks in which memorization is only necessary for \emph{efficient} learners. 
While large memorization (i.e., mutual information) implies that the algorithm is not differentially private, it does not imply the existence of an efficient attack. Indeed, except for \cite{BZ16,ADHLR24}, these works are typically not constructive (i.e., they do not provide an efficient way to translate the memorization into an efficient privacy attack).

\paragraph{Computational Differential Privacy}

\citet{BNO08,MPRV09} considered a computational relaxation of (the standard, information-theoretic) differential privacy, where they require that the outputs of two executions on neighboring datasets are indistinguishable only from the eyes of an \emph{efficient} observer. This relaxation turns out to be necessary for fundamental distributed tasks \cite{McGregorMPRTV10,HaitnerMST22} but also for some (artificial) centralized ones \cite{BunCV16,GhaziIKKM23}. We note that an efficient privacy attack (e.g., membership inference or reconstruction) implies that the algorithm is not computationally differentially private. But in the opposite direction, if an algorithm is not computationally differentially private, the privacy attack, while being efficient, could be very negligible (e.g., the attack might reveal only a single, insignificant, bit of information about a training example, and still violating computational differential privacy).

\paragraph{Membership Inference}
In \emph{membership inference (MI)} \cite{shokri2017membership,yeom2018privacy}, 
an adversary is given a target example and its goal is to infer
whether it was included in a model’s training set. Most techniques follow a paradigm of measuring some correlation/loss function between the target example and the model and checking if it exceeds some threshold. For example, \citet{DworkSSUV15} showed that any algorithm that given $x_1,\ldots,x_n \in \oo^d$ outputs $y \in [-1,1]^d$ that is sufficiently close to the average
$\frac1n\sum_{i=1}^n x_i$, is exposed to the MI attack that given a target example $x$, decides "IN" or "OUT" based on whether $\ip{x,y}$ (the inner-product between $x$ and $y$) exceeds some threshold. The MI literature is very rich, and includes works that perform attacks on fundamental learning tasks (e.g., \cite{SankararamanOJH09,DworkSSUV15,Sablayrolles2019WhiteboxVB,AzizeB24,ADHLR24}) along with MI attacks on more complex models (e.g., \cite{ShokriSSS17,yeom2018privacy,JagielskiTTILCWS23,CarliniCNSTT22,YunhuiEtAl20,Sablayrolles2019WhiteboxVB,WatsonGCS22}). We remark that while MI attacks are indeed a reasonable privacy concern in many scenarios, they are weaker than reconstruction attacks because the attacker needs to know the target training example beforehand while without it there might be no way to extract sensitive information. 

\paragraph{Reconstruction} 

The results of \citet{DiNi03} provided a foundation for rigorously quantifying reconstruction bounds for general query release mechanisms. In their setting, the dataset is binary, and they show that given a few statistical queries, it is possible to reveal $90\%$ of the dataset. \citet{DMT07,DY08,HaitnerMST22} provide improved results in similar settings like \cite{DiNi03} but under weaker accuracy assumptions. In more general settings, \citet{GuoKCV22} demonstrated that both R\'enyi Differential Privacy (DP) and Fisher information offer robust semantic assurances against reconstruction attacks. \citet{balle2022reconstructing} introduced the concept of reconstruction robustness (ReRo), establishing a link between reconstruction attacks and DP. \citet{hayes2023bounding} extended \cite{balle2022reconstructing}'s to analyze reconstruction attacks against DP-SGD. \citet{kaissis2023bounding} furthered this research by examining the connection between the hypothesis testing interpretation of DP (specifically $f$-DP) and reconstruction robustness. From a more practical standpoint, there is a rich literature on reconstruction attacks on various real-world models (e.g., \cite{FredriksonJR15,CarliniCUKS18,YangZCL19,HeTR19,Carlini0EKS19,YinMALMHJK20,CarliniDiffusion23,HaimVYSI22}). In \cref{sec:RecognizingRec} (recognizing reconstruction), we focus on three of them in order to illustrate the expressiveness of our \cref{def:intro:extraction}.

\paragraph{Formalizing legal concepts of privacy} Some of the work towards formalizing legal privacy concepts mathematically has taken an approach similar to our ``sandwiching'' of the concept of reconstructing rather than attempting a precise answer to Question~\ref{q:motivation}. 
These works did not attempt to model a legal concept exactly but rather define requirements which are clearly stricter or clearly weaker than the those of the legal concept, yet non-trivial so they can substantiate a claim that the use of a certain technology satisfies or does not satisfy the legal requirement. 
As one example, in their modeling of the FERPA privacy requirement,\footnote{FERPA -- The Family Educational Rights and Privacy Act (FERPA) -- is is a Federal law that protects the privacy in education records.} \citet{Bridging} provided a definition which is stronger than the legal requirement, yet satisfiable by the use of differential privacy, hence providing strong evidence that the use of differential privacy (with appropriate parameters) satisfies the legal requirements. As another example is the modeling of the GDPR requirement of protection against singling out \citet{CohenN20} defined {\em protection against predicate singling out}, a concept which is weaker than the legal requirement. Showing that $k$-anonymity does not protect against predicate singling out they hence claimed that it does neither satisfy the legal requirement.

\section{Expressiveness of Narcissus resiliency}
\label{sec:NarcissusExpressiveness}

In the introduction we presented the definition of Narcissus resiliency (Definition~\ref{def:narc}) with the interpretation of ``security against reconstruction attacks'', where the function $R$ encapsulates what a ``valid reconstruction'' means. We now show that with different instantiations of the function $R$, the same framework can be used to capture many other existing security notions.  In this section we show this for the notions of {\em membership inference attacks} and for {\em predicate singling out}. In Appendices~\ref{sec:DP-narc} and~\ref{sec:comp-narc} we extend this to additional security notions, including {\em differential privacy}, security of {\em one-way functions}, and security of {\em encryption schemes}.
The fact that Narcissus resiliency is expressive enough to capture all these (seemingly unrelated) well-established notions enhances our confidence in its application as a security notion against reconstruction attacks.

\subsection{Membership inference as Narcissus resiliency}

Membership Inference (MI) attacks \cite{shokri2017membership,yeom2018privacy} are a family of attacks which are applied mostly to machine learning models. 
The goal of an MI attack is to determine whether a given data record was part of the training data underlying the model or not. 
More specifically, let $\DDD$ be a distribution over data records, and let $\MMM$ be an algorithm for analyzing datasets of size $n$. 
Algorithm $\MMM$ is said to be MI-secure if no adversary $\AAA$ has a significant advantage in distinguishing between the outputs of the following two experiments:
\begin{itemize}[leftmargin=20px]
    \item Sample $S\leftarrow\DDD^n$ and $z\leftarrow\DDD$ independently. Let $y\leftarrow\MMM(S)$. Output $(y,z)$.
    \item Sample $S\leftarrow\DDD^n$. Let $z$ be a uniformly random element from $S$. Let $y\leftarrow\MMM(S)$. Output $(y,z)$.
\end{itemize}

The formal definition is as follows.

\begin{definition}[Resilience to Membership Inference, \cite{shokri2017membership,yeom2018privacy}]
Let $\MMM:\XXX^n\rightarrow Y$ be an algorithm that operates on a dataset, and let $\DDD$ be a distribution over $\XXX$. We say that $\MMM$ is $(\delta,\DDD)$-MI-secure if for every adversary $\AAA$ it holds that 
$$
\left|
\underset{\substack{S\leftarrow\DDD^n\\y\leftarrow\MMM(S)\\z\leftarrow\DDD\\b\leftarrow\AAA(y,z)}}{\Pr}[b=1]
-
\underset{\substack{S\leftarrow\DDD^n\\y\leftarrow\MMM(S)\\z\in_{\rm R} S\\b\leftarrow\AAA(y,z)}}{\Pr}[b=1]
\right|
\leq\delta.
$$
\end{definition}

We show that Narcissus resiliency captures the concept of membership inference. We make use of the following function:

\begin{definition}
We define $R_{\rm MI}$ to be a (randomized) binary function, which takes two arguments: a dataset $Z\in\XXX^n$ and a (possibly randomized) function $f:\XXX\rightarrow\{0,1\}$. Given $Z,f$, to compute $R_{\rm MI}(Z,f)$, sample a point $z\in Z$ and return $f(z)$. 
\end{definition}

\begin{theorem}
\label{thm:MInarc} Let $\MMM:\XXX^n\rightarrow Y$ be an algorithm and let $\DDD$ be a distribution over $\XXX$. Then $\MMM$ is $(\delta,\DDD)$-MI-secure if and only if it is $(0,\delta,\{\DDD^n\})$-$R_{\rm MI}$-Narcissus-resilient.
\end{theorem}

\begin{proof}
First observe that an MI-adversary gets two arguments (an outcome $y$ and a point $z$) and returns a bit, while an $R_{\rm MI}$-Narcissus-adversary gets only one argument (the outcome $y$) and returns a binary function $f$ that takes a point and returns a bit (aiming to satisfy $R_{\rm MI}$). To help bridge between these notations, given an MI-adversary $\AAA$ and an outcome $y$, we write $f(\cdot)\leftarrow\AAA(y,\cdot)$ to denote the binary function obtained by fixing $\AAA$'s first argument to be $y$. We can interpret this as an $R_{\rm MI}$-Narcissus-adversary that takes one argument (the outcome $y$) and returns the function $f(\cdot)\leftarrow\AAA(y,\cdot)$.
It follows that 

\begin{align*}
\underset{\substack{S\leftarrow\DDD^n\\y\leftarrow\MMM(S)\\f(\cdot)\leftarrow\AAA(y,\cdot)}}{\Pr}[R_{\rm MI}(S,f)=1]
-
\underset{\substack{S\leftarrow\DDD^n\\T\leftarrow\DDD^n\\y\leftarrow\MMM(S)\\f(\cdot)\leftarrow\AAA(y,\cdot)}}{\Pr}[R_{\rm MI}(T,f)=1]
\quad
=
\quad
\underset{\substack{S\leftarrow\DDD^n\\y\leftarrow\MMM(S)\\z\in_{\rm R} S\\b\leftarrow\AAA(y,z)}}{\Pr}[b=1]
-
\underset{\substack{S\leftarrow\DDD^n\\y\leftarrow\MMM(S)\\z\leftarrow\DDD\\b\leftarrow\AAA(y,z)}}{\Pr}[b=1].
\end{align*}
Hence, if $\MMM$ is $(\delta,\DDD)$-MI-secure then it is also $(0,\delta,\{\DDD^n\})$-$R_{\rm MI}$-Narcissus-resilient. In addition, if the above two differences are positive, then the equality holds also in absolute value, and hence the condition of $(0,\delta,\{\DDD^n\})$-$R_{\rm MI}$-Narcissus-resiliency implies the condition of $(\delta,\DDD)$-MI-security. Otherwise, 
to show that the left hand side cannot be smaller than $-\delta$,
let $\hat{\AAA}$ be the $R_{\rm MI}$-Narcissus-adversary that runs $\AAA$ and returns the ``inverted'' function $\hat{f}\equiv1-f$. Then, 
\begin{align*}
&
\underset{\substack{S\leftarrow\DDD^n\\T\leftarrow\DDD^n\\y\leftarrow\MMM(S)\\f(\cdot)\leftarrow\AAA(y,\cdot)}}{\Pr}[R_{\rm MI}(T,f)=1]
-
\underset{\substack{S\leftarrow\DDD^n\\y\leftarrow\MMM(S)\\f(\cdot)\leftarrow\AAA(y,\cdot)}}{\Pr}[R_{\rm MI}(S,f)=1]\\
&=-\left(
\underset{\substack{S\leftarrow\DDD^n\\T\leftarrow\DDD^n\\y\leftarrow\MMM(S)\\\hat{f}(\cdot)\leftarrow\hat{\AAA}(y,\cdot)}}{\Pr}[R_{\rm MI}(T,\hat{f})=1]
-
\underset{\substack{S\leftarrow\DDD^n\\y\leftarrow\MMM(S)\\\hat{f}(\cdot)\leftarrow\hat{\AAA}(y,\cdot)}}{\Pr}[R_{\rm MI}(S,\hat{f})=1]
\right)\\
&=
\underset{\substack{S\leftarrow\DDD^n\\y\leftarrow\MMM(S)\\\hat{f}(\cdot)\leftarrow\hat{\AAA}(y,\cdot)}}{\Pr}[R_{\rm MI}(S,\hat{f})=1]
-
\underset{\substack{S\leftarrow\DDD^n\\T\leftarrow\DDD^n\\y\leftarrow\MMM(S)\\\hat{f}(\cdot)\leftarrow\hat{\AAA}(y,\cdot)}}{\Pr}[R_{\rm MI}(T,\hat{f})=1]
\leq\delta.
\end{align*}
which follows from the resilience of $\MMM$ against $\hat{\AAA}$.
\end{proof}

\subsection{Predicate singling out as Narcissus resiliency}

\cite{CohenN20} defined a type of privacy attack called {\em Predicate Singling Out (PSO)}, intended to mathematically formulate the legal concept of ``singling out'' that appears in the General Data Protection Regulation (GDPR). Concretely, let $\DDD$ be a data distribution and let $S\leftarrow\DDD^n$ be a dataset containing $n$ iid samples from $\DDD$. Let $\MMM$ be an algorithm that operates on $S$ and return an outcome $y$. A PSO adversary gets $y$ and aims to find a predicate $p$ that ``isolates'' one record in $S$, meaning that it evaluates to 1 on {\em exactly} one record in $S$.  Note, however, that without further restrictions this is not necessarily a hard task. In particular, even without looking at $y$, the adversary might choose a predicate $p$ whose expectation over $\DDD$ is $1/n$. Such a predicate would isolate a record in $S$ with constant probability.  More generally, if the expectation of $p$ is $w$, then the probability that it isolates a record in a fresh dataset of size $n$ is $n\cdot w\cdot(1-w)^{n-1}$.
Thus, for the attack to be considered ``significant'', the adversary has to succeed in its attack with probability noticeably higher than this. The formal definition is as follows.

\begin{definition}[\cite{CohenN20}]\label{def:PSO} Let $\MMM:\XXX^n\rightarrow Y$ be an algorithm that operates on a dataset. We say that $\MMM$ is $(\eps,\delta,w_{\rm max})$-PSO secure if for every $w\leq w_{\rm max}$, every distribution $\DDD$ over $\XXX$, and every adversary $\AAA$ it holds that
$$
\underset{\substack{S\leftarrow\DDD^n\\y\leftarrow\MMM(S)\\ p\leftarrow\AAA(y)}}{\Pr}\left[ \sum_{x\in S}p(x)=1 \wedge  \E_{\DDD}[p]\leq w    \right]
\leq e^{\eps} \cdot 
\sup_{\substack{\text{predicate } p\\ \text{s.t.\ } \E_{\DDD}[p]\leq w}}\left\{  n\cdot \E_{\DDD}[p]\cdot\left(1-\E_{\DDD}[p]\right)^{n-1}  \right\}
+\delta.
$$
\end{definition}

Using the paradigm of Narcissus resiliency, we present an alternative definition for predicate singling out, which is simpler in that it avoids reasoning directly about the expectations of the predicates.

\begin{definition}[Narcissus singling out security]\label{def:NPSO} Let $\MMM:\XXX^n\rightarrow Y$ be an algorithm. Let $R$ be the relation that takes a dataset $S$ and a predicate $p$, where $R(S,p)=1$ if and only if $p$ evaluates to 1 on {\em exactly} one point in $S$. 
We say that $\MMM$ is $(\eps,\delta)$-Narcissus-singling-out-secure if for every distribution $\DDD$ over $\XXX$ and for every attacker $\AAA$ it holds that 
$$\underset{\substack{S\leftarrow\DDD^n\\y\leftarrow\MMM(S)\\p\leftarrow\AAA(y)}}{\Pr}[R(S,p)=1] \leq e^\eps \cdot \underset{\substack{S\leftarrow\DDD^n\\T\leftarrow\DDD^n\\y\leftarrow\MMM(S)\\p\leftarrow\AAA(y)}}{\Pr}[R(T,p)=1] + \delta.$$
\end{definition}

Intuitively, Definitions~\ref{def:PSO} and~\ref{def:NPSO} have similar interpretations: Finding a predicate that isolates a record in $S$ after seeing the outcome of $\MMM$ is almost as hard as doing this without seeing the outcome of $\MMM$. However, the definitions are not equivalent. We leave open the question of understanding the relationships between these two definitions.

\section{Narcissus-Resiliency Prevents Non-Trivial Extraction}\label{sec:NR-prevents-extraction}

Let $\MMM$ be a mechanism which is applied to a dataset $S$ to obtain an outcome $y$. Let $\BBB$ be an ``extraction attacker'' that takes the outcome $y$ and aims to compute a short program $\AAA$ that serves as an extraction evidence for $S$ w.r.t.\ a relation $R$ and a programming language $\cL$. We show that if $\MMM$ is Narcissus-resilient (w.r.t.\ an appropriate reconstruction relation), then it prevents $\BBB$ from succeeding in its attack. Formally,

\begin{definition}\label{def:R:ext}
Fix a parameter $q\leq1$ controlling the desired {\em quality} of extraction (as in Definition~\ref{def:intro:extraction}). 
We define the relation $R_{\rm ext}$ that takes a dataset $S$ and a pair $(y,\AAA)$ where $y$ is an outcome of $\MMM$ and $\AAA$ is a program, and returns 1 if and only if the following two conditions hold:
\begin{enumerate}
        \item $R(S,\AAA(y)) = 1$, and
        \item $\frac{\size{\cA}}{K_{\cL}(\set{z \colon R(S,z)=1})} \leq 1-q $.
\end{enumerate}
\end{definition}

\begin{lemma}[Narcissus-resiliency prevents non-trivial extraction]\label{lemma:NR-prevent-ext-ver2}
Let $\MMM$ be an $(\eps,\delta,\FFF)$-$R_{\rm ext}$-Narcissus-resilient according to \cref{def:narc}, and consider an adversary $\cB$ that takes the outcome of $\MMM$ and outputs an $\cL$-program $\AAA$. Then for every $\DDD\in\FFF$ we have
$$
\underset{\substack{S\leftarrow\DDD\\y\leftarrow\MMM(S)\\\AAA\leftarrow\BBB(y)}}{\Pr}\left[ 
\begin{matrix}
\AAA \text{ is an } (R,\cL)\text{-extraction} \\ \text{evidence of } S \text{ from } y     
\end{matrix}
\right]
\leq
e^{\eps}\cdot
\underset{\substack{S\leftarrow\DDD\\T\leftarrow\DDD\\y\leftarrow\MMM(S)\\\AAA\leftarrow\BBB(y)}}{\Pr}\left[  
\begin{matrix}
\AAA \text{ is an } (R,\cL)\text{-extraction} \\ \text{evidence of } T \text{ from } y     
\end{matrix}
\right]
+\delta
$$
\end{lemma}

\begin{proof}
This follows directly from the assumption that $\MMM$ is $R_{\rm ext}$-Narcissus-resilient. Formally,
let $\hat{\BBB}$ denote the adversary that on input $y$ returns $(y,\BBB(y))$. Then,
\begin{align*}
\underset{\substack{S\leftarrow\DDD\\y\leftarrow\MMM(S)\\\AAA\leftarrow\BBB(y)}}{\Pr}\left[ 
\begin{matrix}
\AAA \text{ is an } (R,\cL)\text{-extraction} \\ \text{evidence of } S \text{ from } y     
\end{matrix}
\right] &=
\underset{\substack{S\leftarrow\DDD\\y\leftarrow\MMM(S)\\(y,\AAA)\leftarrow\hat{\BBB}(y)}}{\Pr}\left[ 
R_{\rm ext}(S,(y,\AAA))=1
\right]\\
&\leq e^{\eps}\cdot \underset{\substack{S\leftarrow\DDD\\T\leftarrow\DDD\\y\leftarrow\MMM(S)\\(y,\AAA)\leftarrow\hat{\BBB}(y)}}{\Pr}\left[ 
R_{\rm ext}(T,(y,\AAA))=1
\right]+\delta\\
&\leq e^{\eps}\cdot
\underset{\substack{S\leftarrow\DDD\\T\leftarrow\DDD\\y\leftarrow\MMM(S)\\\AAA\leftarrow\BBB(y)}}{\Pr}\left[  
\begin{matrix}
\AAA \text{ is an } (R,\cL)\text{-extraction} \\ \text{evidence of } T \text{ from } y     
\end{matrix}
\right]
+\delta
\end{align*}
\end{proof}
\vspace{-10px}
A possible downside of Lemma~\ref{lemma:NR-prevent-ext-ver2} is that it leverages a {\em different} reconstruction relation for Narcissus-resiliency compared to the extraction relation. This was needed in the proof because our definition of extraction considered two conditions: not only that the attacker needs to satisfy the relation $R$, but it needs to do it with a ``short enough'' program. To capture this in the above lemma, we incorporated this condition in the reconstruction relation used for Narcissus-resiliency. This can be avoided in cases where for some parameter $k$ we have:
\begin{enumerate}
    \item The adversary $\BBB$ always outputs a program of length at most $k$; and
    \item The underlying data distribution $\DDD$ is such that with overwhelming probability over sampling $S\leftarrow\DDD$ we have that $\frac{k}{K_{\cL}(\set{z \colon R(S,z)=1})} \leq 1-q$.
\end{enumerate}
That is, in cases when sampling a dataset $S\leftarrow\DDD$ then with overwhelming probability the Kolmogorov complexity $K_{\cL}(\set{z \colon R(S,z)=1})$ is high. Indeed, if this is the case, then the second condition in Definition~\ref{def:R:ext} holds with overwhelming probability, and hence can be ignored, which unifies the two relations $R$ and $R_{\rm ext}$.

\subsection{Limitation of \cref{lemma:NR-prevent-ext-ver2}}\label{sec:limitaion-of-connection}

It is important to note that \cref{lemma:NR-prevent-ext-ver2} is limited to attackers $\cB$ that only see the output $y$ of the mechanism. But some of the real-world attacks also use the training data in order to construct an extraction evidence. E.g., in the Diffusion Model attack of \citet{CarliniDiffusion23} (see \cref{sec:real-world-attacks}), the attacker $\cB$ first processes the input data in order to find captions of ``interesting" images, and then define $\cA$ as the short program that simply queries the model on one of these captions. We note that Narcissus-resiliency does not prevent such attacks in general, because $\cB$ uses the training data to generate $\cA$. But \cref{lemma:NR-prevent-ext-ver2} does imply that it is impossible to achieve non-trivial extraction of training examples, without knowing them in advanced.

	\ifdefined\IsAnonymous
	\else
	\section*{Acknowledgments}

The authors would like to thank Noam Mazor for useful discussions about Kolmogorov complexity.
 
\medskip
\noindent Edith Cohen was partially supported by Israel Science Foundation (grant 1595/19 and 1156/23). Haim Kaplan was partially supported by Israel Science Foundation (grant 1595/19 and 1156/23) and the Blavatnik Family Foundation.
Yishay Mansour's work partially funded from the European Research Council (ERC) under the European Union’s Horizon 2020 research and innovation program (grant agreement No. 882396), by the Israel Science Foundation (grant number 993/17), Tel Aviv University Center for AI and Data Science (TAD), and the Yandex Initiative for Machine Learning at Tel Aviv University.
Shay Moran is a Robert J.\ Shillman Fellow; he acknowledges support by ISF grant 1225/20, by BSF grant 2018385, by an Azrieli Faculty Fellowship, by Israel PBC-VATAT, by the Technion Center for Machine Learning and Intelligent Systems (MLIS), and by the European Union (ERC, GENERALIZATION, 101039692). Views and opinions expressed are however those of the author(s) only and do not necessarily reflect those of the European Union or the European Research Council Executive Agency. Neither the European Union nor the granting authority can be held responsible for them.
Kobbi Nissim's work partially funded by NSF grant No.~2217678 and by a gift to Georgetown University.
Uri Stemmer was Partially supported by the Israel Science Foundation (grant 1419/24) and the Blavatnik Family foundation.
Eliad Tsfadia's work supported by a gift to Georgetown University.
	\fi
	
	\printbibliography
	
	\appendix

\remove{
\section{variant and properties (not goint to stay here...)}

\paragraph{Properties of Narcissus resiliency.} Post processing, composition.

\begin{definition}[variant with minimal success threshold (APPENDIX???)]
Let $\XXX$ be a data domain, let $\FFF$ be a family of distributions over datasets containing elements from $\XXX$, and let $R:\XXX^*\times\{0,1\}^*\rightarrow\{0,1\}$. 
Algorithm $\MMM$ is $(\eps,\delta,\gamma,\FFF)$-$R$-Narcissus-resilient if for all $\DDD\in\FFF$ and for all attackers $\AAA$ it holds that
$$
\underset{\substack{S\leftarrow\DDD\\y\leftarrow\MMM(S)\\z\leftarrow\AAA(y)}}{\Pr}[R(S,z)=1]\leq\max\left\{\gamma\;,\;e^{\eps}\cdot \underset{\substack{S\leftarrow\DDD\\T\leftarrow\DDD\\y\leftarrow\MMM(S)\\z\leftarrow\AAA(y)}}{\Pr}[R(T,z)=1]+\delta\right\}.
$$
\end{definition}
}

\section{Example of a Narcissus Resilient Mechanism that is not Differentially Private}\label{sec:example}

Let $\XXX$ be a data domain, let $\FFF$ be the family of all i.i.d.\ distributions over $\XXX^n$, and let $q \colon \XXX \rightarrow \{0,1\}$ be a counting query. In the following we describe a \emph{deterministic} mechanism that is clearly not differentially private but is $(\eps,\delta,\FFF)$-$R$-Narcissus-resilient for the predicate $R$ that given $S = (x_1,\ldots,x_n)$ and $z=(i,b) \in [n]\times \{0,1, \perp\}$ as inputs, outputs $1$ iff $q(x_i)=b$ (i.e., the adversary's task is to predict the value of the counting query on one data element, and we allow it to ``fail'' intentionally by outputing $\perp$ as a prediction).
The mechanism is describe below.

\begin{algorithm}[Mechanism $\MMM$]\label{alg:M}
	\item Input: $S = (x_1,\ldots,x_n) \in \cX^n$.
	\item Parameters: $\varepsilon, \delta > 0$.
	\item Operation:~
	\begin{enumerate}
		\item Compute $\tilde{p} =  \frac1n\sum_{i=1}^n q(x_i)$ (the empirical mean).
		
		\item Let $\tau = (\gamma^2 + 2\gamma)\cdot \frac{\ln(4/\delta)}{n}$ for $\gamma = \frac{6 e^{\eps}(1-\delta) - 2}{(1-\delta)e^{\eps} - 1}$ (note that $\gamma = \tilde{O}\paren{1/\eps}$ for $\eps < 1$ and $\delta \ll 1$).\label{step:tau}

		\item If $\tilde{p} < \tau$ or $\tilde{p} > 1-\tau$, output $ \perp$ (i.e., abort).
		\item Otherwise, output $\tilde{p}$. 
	\end{enumerate}
\end{algorithm}

Namely, as long as the empirical mean is not too small or too large, the mechanism simply outputs it without any additional noise (that is required for a DP mechanism). We remark that cutting the edges is necessary, as otherwise, the mechanism wouldn't have been Narcissus-resilient.\footnote{Let $\cD = \cP^n$ and assume $p \eqdef \ppr{x \sim \cP}{q(x) = 1} = 1/n^2$. Consider the attacker $\AAA$ that given $\tp$ as input, samples $i \la [n]$ and outputs $(i,1)$ if $\tp > 0$ and otherwise outputs $(i,\perp)$. The probability it reconstructs w.r.t. $S$ is $\pr{\tp > 0} \cdot \Omega(1/n) = \Omega(1/n^2)$, but the probability it reconstructs w.r.t. $T$ (fresh i.i.d.\ samples) is  $\pr{\tp > 0} \cdot 1/n^2 = O(1/n^3)$.} Yet, we show that by cutting it when it is only $\tilde{O}\paren{\frac1{\eps^2 n}}$ close to the edges (for $\eps < 1$ and $\delta \ll 1$) ensures $(\eps,\delta,\FFF)$-$R$-Narcissus-resiliency.

\begin{claim}
        Let $\eps,\delta > 0$ such that $\gamma(\eps,\delta) \geq 6$ and assume that $n \geq 4(\gamma^2 + 2\gamma) \ln(1/\delta)$ (i.e., $\tau \leq 1/4$).
	Then the mechanism $\MMM$ (with parameters $\eps,\delta$) is $(\eps,\delta,\FFF)$-$R$-Narcissus-resilient.
\end{claim}
\begin{proof}
	Fix an i.i.d.\ distribution $\DDD =\cP^n$ and let $p = \ppr{x \sim \cP}{q(x) = 1}$. We assume \wlg that $p \leq 1/2$, as the proof for the case $p \geq 1/2$ holds by symmetry.
 
 Define the random variables $\bS = (\bx_1,\ldots,\bx_n) \sim \cP^n$,  $\bT = (\bx_1',\ldots,\bx_n') \sim \cP^n$,   $\tilde{\bp} =  \frac1n\sum_{i=1}^n q(\bx_i)$ and $\by = \MMM(\bS)$ (note that $\by$ is independent of $\bT$).
	We prove the claim by showing that for every $i \in [n]$:
	\begin{align}\label{example:goal}
		\ppr{y \sim \by}{q(\bx_i)|_{\by = y}  \: \approx_{\eps,\delta/2} \:Bern(p)} \geq 1-\delta/2,
	\end{align}
	where $Bern(p)$ denotes the Bernoulli distribution that outputs $1$ w.p. $p$.
	Given \cref{example:goal} and since $q(\bx_i') \equiv Bern(p)$, we deduce that
	\begin{align}\label{example:indis}
		(\by,q(\bx_i)) \approx_{\eps,\delta} (\by,q(\bx_i')).
	\end{align}
	In order to complete the proof of the claim using (\ref{example:indis}), we now fix an adversary $\AAA\colon ([0,1]\cup \set{\perp}) \rightarrow ([n]\times \set{0,1,\perp})$ which can be described by two (randomized) functions $I,B$ that given $y$ as input, it first samples an index $i \sim I(y) \in [n]$ and then a bit $b \sim B(i,y) \in \set{0,1,\perp}$. For $i \in [n]$ we define the random function $F_i$ that given $x \in \cX$ and $y \in [0,1]$ as inputs, samples $b \sim B(i,y)$ and outputs $\indic{q(x) = b}$. In words, $F_i(x,y)$ emulates an execution of $\AAA(y)$ conditioned on $\AAA$ aiming to predict the \ith bit, and outputs $1$ iff it outputs $x$. Next, define the random variables $\bz = (\bi,\bb) = \Ac(\by)$. Now for every $i \in \Supp(\bi)$, by the definition of $F_i$, it holds that $R(\bS,\bz)|_{\bi=i} \equiv F_i(q(\bx_i),\by)$ and $R(\bT,\bz)|_{\bi=i} \equiv F_i(q(\bx_i'),\by)$. Thus, by \cref{example:indis} and post-processing we deduce that $R(\bS,\bz)|_{\bi=i} \approx_{\eps,\delta} R(\bT,\bz)|_{\bi=i}$. Since the above holds for every $i \in\Supp(\bi)$, we conclude that it also holds for $i \sim \bi$ which yields that $R(\bS,\bz) \approx_{\eps,\delta} R(\bT,\bz)$, as required. \haim{Too much notation in one paragraph, i cannot follow, maybe eliminate unnecessary notation and use words ?}\Enote{I tried to add a bit more explanations, but I'm not sure how to further simplify it.}
	
	It is left to prove \cref{example:goal}. Recall that $\tilde{\bp} = \frac1n \sum_{i=1}^n q(\bx_i)$ where $\sum_{i=1}^n q(\bx_i)$ is distributed as a binomial $B(n,p)$. A standard tail inequality for binomial distributions (\cite{ChungLu}, Lemma 2.1) implies that\haim{which binomial ? what is $\tilde{p}$ ?}\Enote{Is it clearer?}
	\begin{align}\label{eq:binom}
		\forall t>0: \quad \pr{\size{\tilde{\bp} - p} > t} \leq 2\cdot e^{-\frac{t^2 n}{2(p + t/3)}}.
	\end{align}

	
	In the following, let $\Delta = 2\gamma \cdot \frac{\ln(4/\delta)}{n}$ and recall that $\tau =\gamma^2 \cdot \frac{\ln(4/\delta)}{n} + \Delta$ (\stepref{step:tau} of \cref{alg:M}). 
	We now split into three cases.
	
	\paragraph{The case $p \in  [0,\tau - \Delta]$.}
	By \cref{eq:binom} it holds that
	\begin{align*}
	\pr{\size{\tilde{\bp}-p} > \Delta} \leq 2\cdot e^{-\frac{\Delta^2 n}{2(p + \Delta/3)}} \leq 2\cdot e^{-\frac{4 \gamma^2 \ln^2(4/\delta)/n}{2(\gamma^2 \ln(4/\delta)/n + 2\gamma \ln(4/\delta)/(3n))}} \leq 2\cdot e^{-\ln(4/\delta)} \leq \delta/2,
	\end{align*}
    where in the penultimate inequality we used the assumption $\gamma \geq 6$. 
 \haim{maybe make the substitutions explicit so it is easier to follow without a pen and paper}\Enote{Done}
	Therefore, we deduce that
	$\pr{\by = \perp} \geq 1-\delta/2$. The latter implies that $\pr{q(\bx_i) = 1 \mid \by = \perp} \in p \pm \delta/2$ which yields that $q(\bx_i)|_{ \by = \perp} \approx_{0,\delta/2} Bern(p)$.

	\paragraph{The case $p \in [\tau + 2\Delta, \: 1/2]$.}
        Let $\lambda_p \eqdef 2\cdot \sqrt{\frac{p \cdot \ln(4/\delta)}{n}}$ and observe that $\lambda_{\tau-\Delta} = \Delta$ and $\Delta < \lambda_{\tau+\Delta} \leq 2\Delta$. In addition, since $p \geq \tau + 2\Delta$ it holds that $\lambda_p < p$. Thus \cref{eq:binom} implies that
        \begin{align}\label{eq:large-p-conc}
	\pr{\size{\tilde{\bp}-p} > \lambda_p} \leq 2\cdot e^{-\frac{\lambda_p^2 n}{2(p + \lambda_p/3)}} \leq 2\cdot e^{-\frac{4 p \cdot \ln(4/\delta)}{4p}} \leq \delta/2.
	\end{align}
        Furthermore, note that for every $\tp \in p \pm \lambda_p$ it holds that
	\begin{align}\label{eq:tau-delta-calc}
			\max\set{\frac{p}{\tp}, \frac{\tp}{p}} 
                &\leq \frac{\max\set{p,\tp}}{\min\set{p,\tp}} 
			\leq \frac{p}{p - \lambda_p}
			\leq \frac{\tau + \Delta}{\tau + \Delta - \lambda_{\tau+\Delta}}
                \leq \frac{\tau + \Delta}{\tau - \Delta}\nonumber\\
			&\leq \frac{\tau - 2\Delta}{(1-\delta)(\tau - 4\Delta)}
			= \frac{\gamma - 2}{(1-\delta)(\gamma - 6)}
			= \frac{ \frac{6 e^{\eps}(1-\delta) - 2}{(1-\delta)e^{\eps} - 1} - 2}{(1-\delta)( \frac{6 e^{\eps}(1-\delta) - 2}{(1-\delta)e^{\eps} - 1} - 6)}
			= e^{\eps},
		\end{align}
        where the second, third, and fifth inequalities hold since for every $a > b > c > 0$ we have $\frac{a}{b} \leq \frac{a-c}{b-c}$. 
	
	We deduce that for every $\tp \in p \pm \lambda_p$ it holds that $q(\bx_i)|_{\by = \tp} \approx_{\eps,0} Bern(p)$, and thus conclude that
	\begin{align*}
			\ppr{y \sim \by}{q(\bx_i)|_{\by = y}  \: \approx_{\eps,0} \:Bern(p)}
			&\geq (1 - \delta/2) \ppr{y \sim \by|_{\size{\tilde{\bp} - p} \leq \lambda_p}}{q(\bx_i)|_{\by = y}  \: \approx_{\eps,0} \:Bern(p)}\\
			&= (1 - \delta/2) \ppr{\tp \sim \tilde{\bp}|_{\size{\tilde{\bp} - p} \leq \lambda_p}}{q(\bx_i)|_{\by = \tp}  \: \approx_{\eps,0} \:Bern(p)}\\
			&= 1-\delta/2,
		\end{align*}
    where the inequality holds by \cref{eq:large-p-conc} and the first equality holds since $\size{\tilde{\bp} - p} \leq \lambda_p$ implies that $\by = \tilde{\bp}$ because the function $f(p) = p - \lambda_p$ is monotonic increasing and $f(\tau + 2\Delta) > 0$.

		\paragraph{The case $p \in [\tau - \Delta, \tau + 2\Delta]$.}

            Here we use a truncated version of Lemma 2.1 (inequality 2.1) in \cite{ChungLu}, which implies that:
            \begin{align}\label{CL:truncated}
                \forall \: 0 \leq a \leq b \leq p: \quad \pr{\tilde{\bp} < p - b \mid \tilde{\bp} < p - a} \leq 2\cdot e^{-\frac{(b-a)^2 n}{2p}}.
            \end{align}
            If $p \in [\tau, \tau + 2\Delta]$ we deduce by \cref{CL:truncated} that
            \begin{align*}
                \pr{\tilde{\bp}  < \tau - 2\Delta \mid \tilde{\bp} < \tau} \leq 2\cdot e^{-\frac{4\Delta^2 n}{2 p}} \leq e^{-\frac{4\Delta^2 n}{2 (\tau + 2\Delta)}}= 2\cdot e^{-\frac{16\gamma^2 \ln^2(4/\delta)}{2 (\gamma^2 + 5\gamma) \ln(4/\delta)}} \leq \delta.
            \end{align*}
            If $p \in [\tau - \Delta, \tau]$ then by \cref{CL:truncated} we have
            \begin{align*}
                \pr{\tilde{\bp}  < \tau - 2\Delta \mid \tilde{\bp} < \tau} = \frac{\pr{\tilde{\bp}  < \tau - 2\Delta}}{\pr{\tilde{\bp}  < \tau}} \leq 2\cdot \pr{\tilde{\bp}  < \tau - 2\Delta} \leq 2\cdot e^{-\frac{\Delta^2 n}{2\tau}} = 2\cdot e^{-\frac{4\gamma^2 \ln^2(4/\delta)}{2(\gamma^2+\gamma) \ln(4/\delta)}} \leq \delta.
            \end{align*}
  
		Overall, in both cases, we conclude that
		\begin{align*}
				\pr{q(\bx_i)=1 \mid \by = \perp} = \ex{\tilde{\bp} \mid \tilde{\bp} < \tau} \geq (1-\delta) (\tau - 2\Delta).
			\end{align*}
		This implies that $q(\bx_i)|_{\by = \perp}  \: \approx_{\eps,0} \:Bern(p)$ since
		\begin{align*}
				\frac{p}{\pr{q(\bx_i)=1 \mid \by = \perp}}
				\leq \frac{p}{(1-\delta) (\tau - 2\Delta)}
				\leq  \frac{\tau-\Delta}{(1-\delta) (\tau - 2\Delta)}
                    \leq  \frac{\tau-2\Delta}{(1-\delta) (\tau - 4\Delta)}
				    = e^{\eps}.
			\end{align*}
		where the last inequality holds by \cref{eq:tau-delta-calc}.

            Furthermore, for every $\tp \in p \pm 2\Delta$ we have that $q(\bx_i)|_{\by = \tp}  \: \approx_{\eps,0} \:Bern(p)$ since
            \begin{align*}
                \max\set{\frac{p}{\tp}, \frac{\tp}{p}} 
                \leq \frac{\max\set{p,\tp}}{\min\set{p,\tp}} 
			\leq \frac{p}{p - 2\Delta}
                \leq \frac{\tau - 2\Delta}{\tau - 4\Delta}
                \leq e^{\eps}
            \end{align*}
            where the second and third inequalities hold since for every $a > b > c > 0$ we have $\frac{a}{b} \leq \frac{a-c}{b-c}$, and the last one holds by \cref{eq:tau-delta-calc}.

            We thus conclude that
            \begin{align*}
            \ppr{y \sim \by}{q(\bx_i)|_{\by = y}  \: \approx_{\eps,0} \:Bern(p)}
            &\geq  \pr{\size{p - \tilde{\bp}}\leq 2\Delta}\cdot \ppr{y \sim \by|_{\size{p - \tilde{\bp}}\leq 2\Delta}}{q(\bx_i)|_{\by = y}  \: \approx_{\eps,0} \:Bern(p)}\\
            &= \pr{\size{p - \tilde{\bp}}\leq 2\Delta}\cdot 1\\
            &\leq 2\cdot e^{-\frac{4\Delta^2 n}{2(p + 2\Delta/3)}}
            \leq 2\cdot e^{-\frac{4\Delta^2 n}{2(\tau + 2\Delta + 2\Delta/3)}}
            \leq 2\cdot e^{-\frac{4\gamma^2 \ln^2(4/\delta)}{2(\gamma^2 + 6\gamma) \ln(4/\delta)}}
            \leq \delta/2,
            \end{align*}
            as required. The second inequality holds by \cref{eq:binom}, the third one holds since $p \leq \tau + 2\Delta$, and the last one holds since $\gamma \geq 6$.

\end{proof}

\section{Differential privacy as Narcissus resiliency}\label{sec:DP-narc}

Differential privacy is a mathematical framework for ensuring that algorithms that analyze data protect the privacy of individual-level information within their input datasets. Informally, an algorithm  is said to be {\em differentially private} if adding or removing the data of one individual from the input dataset has almost no effect on the behaviour of the algorithm. Therefore, the risk incurred by participation in the data is low. Formally:

\begin{definition}[Differential Privacy \cite{DMNS06}]
Algorithm $\MMM:\XXX^n\rightarrow Y$ is {\em $(\eps,\delta)$-differentially private} if for all $S,S'\in \XXX^n$ that differ on a single entry (such datasets are called {\em neighboring}) and for every subset $T\subseteq Y$ it holds that
$$\underset{y\leftarrow \MMM(S)}{\Pr}[y\in T] \leq e^\eps \cdot \underset{y\leftarrow \MMM(S')}{\Pr}[y\in T] + \delta.$$
\end{definition}

In Appendix~\ref{appendix:dp}, we demonstrate that differential privacy can be expressed within the framework of Narcissus resiliency (as defined in Definition~\ref{def:narc}); however, this translation does incur some loss in terms of the resulting parameters.
To overcome this, we introduce a slight variant of Narcissus resiliency that differs from Definition~\ref{def:narc} in that the family of distributions $\FFF$ is of distributions over {\em pairs} of datasets: \haim{As I asked before, do we need the family $\FFF$}\Unote{This is more general. It allows you to talk about algorithms which are secure not just w.r.t.\ a single data distribution. E.g., secure w.r.t.\ every iid distribution.}
\begin{definition}[Narcissus resiliency, variant]\label{def:narcDP}
Let $\XXX$ be a data domain, let $\FFF\subseteq\Delta(\XXX^n\times\XXX^n)$ be a family of distributions over pairs of datasets containing elements from $\XXX$, and let $R:\XXX^*\times\{0,1\}^*\rightarrow\{0,1\}$. 
Algorithm $\MMM$ is $(\eps,\delta,\FFF)$-$R$-Narcissus-resilient if for all $\DDD\in\FFF$ and for all attackers $\AAA$ it holds that
$$
\underset{\substack{(S_0,S_1)\leftarrow\DDD\\b\leftarrow\{0,1\}\\y\leftarrow\MMM(S_b)\\z\leftarrow\AAA(y)}}{\Pr}[R(S_b,z)=1]\leq e^{\eps}\cdot \underset{\substack{(S_0,S_1)\leftarrow\DDD\\b\leftarrow\{0,1\}\\y\leftarrow\MMM(S_b)\\z\leftarrow\AAA(y)}}{\Pr}[R(S_{1-b},z)=1]+\delta.
$$
\end{definition}

\haim{This entire discussion seems rather technical and not so interesting.. in particular we do not say what is the "loss" that we try to overcome with the additional complication}\Unote{When presenting a new concept, you never really know if it is the right concept or not. The point with all these examples is to increase our confidence that it is right: If it captures as a special case many unrelated notions, then to me this increases my confidence in it.}

This allows us to capture the concept of differential privacy without losing (almost) anything in terms of the resulting parameters. Specifically,

\begin{definition}\label{def:DP-pairs}
For a pair of datasets $S,T$ we write $\DDD_{S,T}$ to denote the constant distribution that returns the pair $(S,T)$ with probability 1.
\end{definition}

\begin{theorem}
\label{thm:DPnarc} 
Let $\MMM:\XXX^n\rightarrow Y$ be an algorithm and let $\FFF=\{\DDD_{S,T} : S,T\in\XXX^n \text{ are neighboring datasets}\}$. Then,
\begin{enumerate}
    \item[(1)] If $\MMM$ is $(\eps,\delta)$-differentially private then $\MMM$ is $(\eps,\delta,\FFF)$-$R$-Narcissus-resilient for all $R$.
    \item[(2)]  If $\MMM$ is not $(\eps,\delta)$-differentially private then there exists $R$ such that $\MMM$ is not $(\eps,\frac{\delta}{2},\FFF)$-$R$-Narcissus-resilient for $R$.
\end{enumerate}
\end{theorem}

Note that the above theorem refers to the variant of Narcissus resiliency as in Definition~\ref{def:narcDP}.

\begin{proof}
Suppose that $\MMM$ is $(\eps,\delta)$-DP. Fix a distribution $\DDD=\DDD_{S_0,S_1}\in\FFF$, let $R$ be a relation, and let $\AAA$ be an attacker. We calculate,  \haim{why does the equality below hold ?}\Unote{Just replace $b$ with $1-b$, which changes nothing...}
\begin{align*}
\underset{\substack{(S_0,S_1)\leftarrow\DDD\\b\leftarrow\{0,1\}\\y\leftarrow\MMM(S_b)\\z\leftarrow\AAA(y)}}{\Pr}[R(S_b,z)=1]&
\leq
e^{\eps}\cdot
\underset{\substack{(S_0,S_1)\leftarrow\DDD\\b\leftarrow\{0,1\}\\y\leftarrow\MMM(S_{1-b})\\z\leftarrow\AAA(y)}}{\Pr}[R(S_b,z)=1]+\delta
=
e^{\eps}\cdot
\underset{\substack{(S_0,S_1)\leftarrow\DDD\\b\leftarrow\{0,1\}\\y\leftarrow\MMM(S_b)\\z\leftarrow\AAA(y)}}{\Pr}[R(S_{1-b},z)=1]+\delta
\end{align*}

This concludes the analysis for the first direction. Now suppose that $\MMM$ is not $(\eps,\delta)$-DP. There exist neighboring datasets $S_0,S_1$ and event $B$ such 
that 
$$
\Pr[\MMM(S_0)\in B]>e^{\eps}\cdot\Pr[\MMM(S_1)\in B]+\delta.
$$

We need to show that there is a relation $R$, a distribution in $\FFF$, and an adversary $\AAA$ for which the condition of Narcissus Resiliency does not hold. We let $\AAA$ be the identify function and let $\DDD=\DDD_{S_0,S_1}$. Now let $R$ be such that $R(S_b,y)=1$ if and only if $S_b=S_0$ and $y\in B$. With these notations, to contradict Narcissus Resiliency, we need to show that 
$$\underset{\substack{(S_0,S_1)\leftarrow\DDD\\b\leftarrow\{0,1\}\\y\leftarrow\MMM(S_b)\\y\leftarrow\AAA(y)}}{\Pr}[R(S_b,y)=1]
\gg
\underset{\substack{(S_0,S_1)\leftarrow\DDD\\b\leftarrow\{0,1\}\\y\leftarrow\MMM(S_b)\\y\leftarrow\AAA(y)}}{\Pr}[R(S_{1-b},y)=1].$$
We now calculate each of these two probabilities:
\begin{align*}
\underset{\substack{(S_0,S_1)\leftarrow\DDD\\b\leftarrow\{0,1\}\\y\leftarrow\MMM(S_b)\\y\leftarrow\AAA(y)}}{\Pr}[R(S_b,y)=1] & = \frac{1}{2}\cdot\left(\underset{\substack{y\leftarrow\MMM(S_0)\\y\leftarrow\AAA(y)}}{\Pr}[R(S_0,y)=1]+\underset{\substack{y\leftarrow\MMM(S_1)\\y\leftarrow\AAA(y)}}{\Pr}[R(S_1,y)=1]\right)\\
& = \frac{1}{2}\cdot\left(
\Pr[\MMM(S_0)\in B]
+0\right)
=\frac{1}{2}\cdot\Pr[\MMM(S_0)\in B].
\end{align*}

On the other hand,

\begin{align*}
\underset{\substack{(S_0,S_1)\leftarrow\DDD\\b\leftarrow\{0,1\}\\y\leftarrow\MMM(S_b)\\y\leftarrow\AAA(y)}}{\Pr}[R(S_{1-b},y)=1]
&=\Pr[b=0]\cdot\underset{\substack{y\leftarrow\MMM(S_0)\\y\leftarrow\AAA(y)}}{\Pr}[R(S_1,y)=1]
+
\Pr[b=1]\cdot\underset{\substack{y\leftarrow\MMM(S_1)\\y\leftarrow\AAA(y)}}{\Pr}[R(S_0,y)=1]\\
&=\frac{1}{2}\cdot0
+\frac{1}{2}\cdot \Pr[\MMM(S_1)\in B]\\
&<\frac{1}{2}\left(e^{-\eps}\cdot \Pr[\MMM(S_0)\in B]-\delta\cdot e^{-\eps}\right)\\
%
%
&=e^{-\eps}\left(\underset{\substack{(S_0,S_1)\leftarrow\DDD\\b\leftarrow\{0,1\}\\y\leftarrow\MMM(S_b)\\y\leftarrow\AAA(y)}}{\Pr}[R(S_b,y)=1]-\frac{\delta}{2}\right).
\end{align*}
\end{proof}
%

\section{Computational Narcissus resiliency}\label{sec:comp-narc}
Towards capturing computational concepts such as the security of encryption schemes and one-way functions, we introduce the following computational variant of Narcissus resiliency:

\begin{definition}[Computational Narcissus resiliency]\label{def:compNarc}
Let $\lambda$ be a security parameter, let $\eps>0$, let $\delta:\N \times \N \rightarrow [0,1]$ be a function, and let $n=\poly(\lambda)$.  
Let $\FFF_\lambda \subseteq\Delta((\{0,1\}^*)^{n}\times(\{0,1\}^*)^{n})$ be an ensemble of families of distributions over pairs of datasets,
and let $R:(\{0,1\}^*)^*\times\{0,1\}^*\rightarrow\{0,1\}$ be a relation.  
Algorithm $\MMM$ is $(\eps,\delta,\FFF)$-$R$-computational-Narcissus-resilient if for all $\DDD\in\FFF$ and for all (non-uniform) attackers $\AAA$ running in $\poly(\lambda)$ time it holds that
$$
\underset{\substack{(S_0,S_1)\leftarrow\DDD\\b\leftarrow\{0,1\}\\y\leftarrow\MMM(1^\lambda,S_b)\\z\leftarrow\AAA(y)}}{\Pr}[R(S_b,z)=1]\leq e^{\eps}\cdot \underset{\substack{(S_0,S_1)\leftarrow\DDD\\b\leftarrow\{0,1\}\\y\leftarrow\MMM(1^\lambda,S_b)\\z\leftarrow\AAA(y)}}{\Pr}[R(S_{1-b},z)=1]+\delta(n,\lambda)+\negl(\lambda).
$$
\end{definition}

\begin{remark}
We are generally interested in Definition~\ref{def:compNarc} only under the restrictions that $\MMM$ and $R$ are computationally efficient, and every distribution in $\FFF$ can be sampled in $\poly(\lambda)$ time.
\end{remark}

\subsection{One-wayness as computational Narcissus resiliency}

A {\em one-way function} is a function $f$ that is easy to compute but computationally difficult to invert.  This means that given an input $x$, it is easy to calculate the corresponding output $f(x)$, but given an output $y$ it is extremely hard to find an input $x$ such that $f(x)=y$. Such functions are a foundational concept in cryptography, which are used in various applications like hashing, pseudorandom generators, and digital signatures. Formally,

\begin{definition}[\cite{diffie1976new,yao1982theory}]
Let 
$f:\{0,1\}^*\rightarrow\{0,1\}^*$ be an efficiently computable function. Let $\lambda\in\N$ be a parameter and let $\UUU_{\lambda}$ be the uniform distribution over $\{0,1\}^{\lambda}$. The function $f$ is {\em one-way} if for every (non-uniform) adversary $\AAA$ running in $\poly(\lambda)$ time it holds that
$$
\underset{\substack{y\leftarrow f(\UUU_\lambda) \\ z\leftarrow\AAA(y)}}{\Pr}[f(z)=y] \leq \negl(\lambda).
$$
\end{definition}

We show that this can be captured by the paradigm of Narcissus resiliency. For technical reasons, we will only show this for functions $f$ which are {\em well spread}, which intuitively means that the outcome distribution of $f(x)$ on a random $x$ is ``close to uniform''. Formally,

\begin{definition}
Let $f:\{0,1\}^*\rightarrow\{0,1\}^*$. 
We say that $f$ is $\nu(\lambda)$-well-spread if for every $\lambda\in\N$ and every $y\in\{0,1\}^*$ it holds that
$\Pr_{x\sim\{0,1\}^\lambda}[f(x)=y]\leq\nu(\lambda).$ 
If $f$ is $\nu(\lambda)$-well-spread for a negligible $\nu(\lambda)$ we say that $f$ is well-spread.
\end{definition}

Note that every 1-1 function is well-spread. We prove the following theorem.

\begin{theorem}
\label{thm:OnewayNarc} Let 
$f:\{0,1\}^*\rightarrow\{0,1\}^*$ be an efficiently computable well-spread function. Let $\MMM_f$ be the algorithm that takes an input $x$ and returns $f(x)$. 
Let $R:\{0,1\}^*\times\{0,1\}^*\rightarrow\{0,1\}$ be s.t.\ $R(x,x')=1$ iff $f(x)=f(x')$. 
Let $\lambda\in\N$ be a parameter and let $\UUU_{\lambda}$ be the uniform distribution over $\{0,1\}^{\lambda}$.
Then, $f$ is a one-way function if and only if $\MMM_f$ is $(0,\negl(\lambda),\{(\UUU_{\lambda},\UUU_{\lambda})\})$-$R$-computationally-Narcissus-resilient. 
\end{theorem}

\begin{proof}
First, since $\DDD=(\UUU_{\lambda},\UUU_{\lambda})$, note that for any function $f$ it holds that
\begin{equation}\label{eq:oneway}
\underset{\substack{(S_0,S_1)\leftarrow\DDD\\b\leftarrow\{0,1\}\\y\leftarrow\MMM_f(S_b)\\z\leftarrow\AAA(y)}}{\Pr}[R(S_b,z)=1] 
= \underset{\substack{y\leftarrow f(\UUU_\lambda) \\ z\leftarrow\AAA(y)}}{\Pr}[f(z)=y].    
\end{equation}
In addition, 
\begin{align}
\underset{\substack{(S_0,S_1)\leftarrow\DDD\\b\leftarrow\{0,1\}\\y\leftarrow\MMM_f(S_b)\\z\leftarrow\AAA(y)}}{\Pr}[R(S_{1-b},z)=1]
&=\underset{\substack{(S_0,S_1)\leftarrow\DDD\\y\leftarrow\MMM_f(S_0)\\z\leftarrow\AAA(y)}}{\Pr}[R(S_{1},z)=1]\nonumber\\
&=
\underset{\substack{S_0\leftarrow\UUU_{\lambda}\\y\leftarrow\MMM_f(S_0)\\z\leftarrow\AAA(y)}}{\E}
\underset{\substack{S_1\leftarrow\UUU_{\lambda}}}{\E}
[R(S_{1},z)]\nonumber\\
&=
\underset{\substack{S_0\leftarrow\UUU_{\lambda}\\y\leftarrow\MMM_f(S_0)\\z\leftarrow\AAA(y)}}{\E}
\underset{\substack{S_1\leftarrow\UUU_{\lambda}}}{\Pr}
[f(S_1)=f(z)]\nonumber\\
&\leq\nu(\lambda),\label{eq:oneway2}
\end{align}
where the last inequality follows as $f$ is $\nu$-well-spread. 
Now observe that if $f$ is one-way then the probability stated in (\ref{eq:oneway}) is negligible, and hence the condition of computational-Narcissus-resiliency holds. In the other direction, if $\MMM_f$ is computational-Narcissus-resilient, then the probability stated in (\ref{eq:oneway}) is upper bounded by the expression stated in (\ref{eq:oneway2}), which proves that $f$ is a one-way function.
\end{proof}

\subsection{Ciphertext Indistinguishability (symmetric key) as Narcissus resiliency}

Private-key encryption, also known as symmetric encryption, is a cryptographic method where the same secret key is used for both encrypting and decrypting information. Formally,

\begin{definition}[private-key encryption scheme; adapted from \cite{KatzLindell2014}]
A {\em private-key encryption scheme} is a tuple of probabilistic polynomial-time algorithms $(\Gen, \Enc, \Dec)$ such that:
\begin{enumerate}
\item  The {\em key-generation algorithm} $\Gen$ takes as input $1^{\lambda}$ and outputs a key $k$. 

\item The {\em encryption algorithm} $\Enc$ takes as input a key $k$ and a plaintext
message $m\in\{0, 1\}^\lambda$, and outputs a ciphertext $c$.

\item  The {\em decryption algorithm} $\Dec$ takes as input a key $k$ and a ciphertext $c$, and outputs a message $m$ or an error symbol $\bot$. 
\end{enumerate}
It is required that for every $\lambda$, every key $k$ output by $\Gen(1^\lambda)$, and every $m \in\{0, 1\}^\lambda$, it holds that $\Dec_k(\Enc_k(m)) = m$.
\end{definition}

Informally, the security requirement is that ciphertexts generated by $\Enc$ are ``unreadable'' without the key. There are several ways to formulate this (capturing different type of attackers). The following is perhaps the most basic security definition.

\begin{definition}[indistinguishable encryptions -- private key; adapted from \cite{KatzLindell2014}]\label{def:EAV-security}
A private-key encryption scheme $\Pi=(\Gen,\Enc,\Dec)$ has {\em indistinguishable encryptions in the presence of an eavesdropper}, or is {\em EAV-secure}, if for 
every polynomial non-uniform adversary $\AAA$ there is a negligible function $\negl=\negl(\lambda)$ such that for every $\lambda$ and for every pair of messages $m_0,m_1\in\{0,1\}^\lambda$ we have that
$$
\underset{\substack{
k\leftarrow\Gen(1^{\lambda})\\
b\in_R\{0,1\}\\
c\leftarrow\Enc_k(m_b)\\
\hat{b}\leftarrow\AAA(c)
}}{\Pr}[\hat{b}=b] \leq \frac{1}{2} + \negl(\lambda).
$$
\end{definition}

We show that EAV-security can be stated via the  paradigm of Narcissus resiliency:

\begin{theorem}
\label{thm:EncNarc} 
Let $\Pi=(\Gen,\Enc,\Dec)$ be a private-key encryption scheme. Let $\MMM_{\Pi}$ be the algorithm that takes a security parameter $\lambda$ and a message $(b\circ m)\in\{0,1\}\times\{0,1\}^{\lambda}$, samples a key $k\leftarrow\Gen(1^\lambda)$, encrypts $c\leftarrow\Enc_k(m)$ and returns $c$. Also let $R$ be the relation that takes two strings and returns 1 if and only if they agree on their first bit. 
For $m_0,m_1\in\{0,1\}^\lambda$ let $\DDD_{m_0,m_1}$  be the (constant) distribution that returns the pair $(0 m_0,1 m_1)$ with probability 1.
Let $\FFF=\{\FFF_{\lambda}\}$ be the ensemble of families of distributions  $\FFF_{\lambda}=\{\DDD_{m_0,m_1}\, :\, m_0,m_1\in\{0,1\}^\lambda\}$. Then, 
$\Pi$ is EAV-secure if and only if $\MMM_{\Pi}$ is $(0,0,\FFF)$-$R$-computational-Narcissus-resilient.
\end{theorem}

\begin{proof}
Fix an adversary $\AAA$, a security parameter $\lambda$, and a pair of messages $m_0,m_1\in\{0,1\}^{\lambda}$. 
First note that
\begin{align*}
\underset{\substack{(0m_0,1m_1)\leftarrow\DDD\\b\leftarrow\{0,1\}\\c\leftarrow\MMM_{\Pi}(1^\lambda,b m_b)\\ \hat{b}\leftarrow \AAA(c)}}{\Pr}[R(b m_b,\hat{b})=1]
=
1-
\underset{\substack{(0m_0,1m_1)\leftarrow\DDD\\b\leftarrow\{0,1\}\\c\leftarrow\MMM_{\Pi}(1^\lambda,b m_b)\\\hat{b}\leftarrow \AAA(c)}}{\Pr}[R(b m_b,\hat{b})=0]
=
1-
\underset{\substack{(0m_0,1m_1)\leftarrow\DDD\\b\leftarrow\{0,1\}\\c\leftarrow\MMM_{\Pi}(1^\lambda,b m_b)\\\hat{b}\leftarrow \AAA(c)}}{\Pr}[R((1-b)m_{1-b},\hat{b})=1].
\end{align*}

Thus, $\MMM_{\Pi}$ is $(0,0,\FFF)$-$R$-computational-Narcissus-resilient if and  only if
\begin{align*}
\underset{\substack{(0m_0,1m_1)\leftarrow\DDD\\b\leftarrow\{0,1\}\\c\leftarrow\MMM_{\Pi}(1^\lambda,b m_b)\\\hat{b}\leftarrow \AAA(c)}}{\Pr}[R(b m_b,\hat{b})=1]\leq\frac{1}{2}+\negl(\lambda),
\end{align*}
which concludes the proof as
\begin{align*}
\underset{\substack{(0m_0,1m_1)\leftarrow\DDD\\b\leftarrow\{0,1\}\\c\leftarrow\MMM_{\Pi}(1^\lambda,b m_b)\\\hat{b}\leftarrow \AAA(c)}}{\Pr}[R(b m_b,\hat{b})=1]=\underset{\substack{
k\leftarrow\Gen(1^{\lambda})\\
b\leftarrow_R\{0,1\}\\
c\leftarrow\Enc_k(m_b)\\
\hat{b}\leftarrow\AAA(c)
}}{\Pr}[\hat{b}=b].
\end{align*}
\end{proof}

\subsection{Ciphertext Indistinguishability (public key) as Narcissus resiliency}

Public-key encryption, also known as asymmetric encryption, is a cryptographic system that uses {\em pairs} of keys: a public key, used for encrypting messages, and a secret key, used for decrypting ciphertexts. Formally,

\begin{definition}[Public-key encryption scheme; adapted from \cite{KatzLindell2014}]
A {\em public-key encryption scheme} is a tuple of probabilistic polynomial-time algorithms $(\Gen, \Enc, \Dec)$ such that:
\begin{enumerate}
\item  The {\em key-generation algorithm} $\Gen$ takes as input $1^{\lambda}$ and outputs a pair of keys $(k_d,k_e)$. Here $k_d$ is referred to as the ``decryption key'' or the ``secret key'', $k_e$ is referred as the ``encryption key'' or the ``public key''.

\item The {\em encryption algorithm} $\Enc$ takes as input an encryption key $k_e$ and a plaintext
message $m\in\{0, 1\}^\lambda$, and outputs a ciphertext $c$.

\item  The {\em decryption algorithm} $\Dec$ takes as input a decryption key $k_d$ and a ciphertext $c$, and outputs a message $m$ or an error symbol $\bot$. 
\end{enumerate}
It is required that for every $\lambda$, every key pair $(k_d,k_e)$ output by $\Gen(1^\lambda)$, and every $m \in\{0, 1\}^\lambda$, it holds that $\Dec_{k_d}(\Enc_{k_e}(m)) = m$.
\end{definition}

Recall that EAV-security (Definition~\ref{def:EAV-security}) states that, for any two messages $m_0$ or $m_1$, an adversary that sees a ciphertext $c$ (and nothing else) cannot distinguish between whether this ciphertext encrypts the message $m_0$ or $m_1$. 
The following definition extends this to the public-key setting, where the adversary has access to the public key (which effectively gives it access to an encryption oracle for free).

\begin{definition}[indistinguishable encryptions -- public key; adapted from \cite{KatzLindell2014}]\label{def:indPK}
A public-key encryption scheme $\Pi=(\Gen,\Enc,\Dec)$ has {\em 
indistinguishable encryptions under a chosen-plaintext attack}, or is {\em CPA-secure}, if for 
every polynomial non-uniform adversary $\AAA=(\BBB,\CCC)$ there is a negligible function $\negl=\negl(\lambda)$ such that for every $\lambda$ we have that
$$
\underset{\substack{
(k_d,k_e)\leftarrow\Gen(1^{\lambda})\\
b\leftarrow_R\{0,1\}\\
(m_0,m_1,{\sf state})\leftarrow \BBB(k_e) \\\
c\leftarrow\Enc_{k_e}(m_b)\\
\hat{b}\leftarrow\CCC({\sf state},c)
}}{\Pr}[\hat{b}=b] \leq \frac{1}{2} + \negl(\lambda).
$$
\end{definition}

Note that in the above definition, the adversary is an {\em interactive protocol}, rather then a single-shot computation, in the sense that if first gets the public-key $k_e$, then it chooses the two messages $m_0,m_1$, then it gets the ciphertext $c$, and finally it outputs the bit $\hat{b}$. Thus, towards capturing CPA-security under the umbrella of Narcissus resiliency, we first extend the concept of Narcissus resiliency to {\em interactive protocols}.

\begin{definition}[Computational Narcissus resiliency of protocols]\label{def:protocolNarc}
Let $\lambda$ be a security parameter, let $\eps>0$, let $\delta:\N \times \N \rightarrow [0,1]$ be a function, and let $n=\poly(\lambda)$.  
Let $\FFF_\lambda \subseteq\Delta((\{0,1\}^*)^{n}\times(\{0,1\}^*)^{n})$ be an ensemble of families of distributions over pairs of datasets,
and let $R:(\{0,1\}^*)^*\times\{0,1\}^*\rightarrow\{0,1\}$ be a relation.  
Protocol $\MMM$ is $(\eps,\delta,\FFF)$-$R$-computational-Narcissus-resilient if for all $\DDD\in\FFF$ and for all (non-uniform) attackers $\AAA$ running in $\poly(\lambda)$ time it holds that
$$
\underset{\substack{(S_0,S_1)\leftarrow\DDD\\b\leftarrow\{0,1\}\\z\leftarrow\left(\AAA\rightleftarrows\MMM(1^\lambda,S_b)\right)}}{\Pr}[R(S_b,z)=1]
\leq e^{\eps}\cdot \underset{\substack{(S_0,S_1)\leftarrow\DDD\\b\leftarrow\{0,1\}\\
z\leftarrow\left(\AAA\rightleftarrows\MMM(1^\lambda,S_b)\right)
}}{\Pr}[R(S_{1-b},z)=1]+\delta(n,\lambda)+\negl(\lambda).
$$
\end{definition}

We now show that the above definition captures CPA-security as a special case.

\begin{theorem}
\label{thm:PubEncNarc} 
Let $\Pi=(\Gen,\Enc,\Dec)$ be a public-key encryption scheme. Let $\MMM_{\Pi}$ be the protocol that takes a bit $b$ and a security parameter $\lambda$, samples a key pair $(k_d,k_e)\leftarrow\Gen(1^\lambda)$, outputs $k_e$, obtains a pair of messages $(m_0,m_1)\in\{0,1\}^{\lambda}$, encrypts $c\leftarrow\Enc_k(m_b)$ and returns $c$. Also let $R$ be the relation that takes two bits and returns 1 if and only if they are equal. Let $\DDD$ denote the trivial distribution that returns $(0,1)$ with probability 1, and denote $\FFF=\{\DDD\}$. Then, 
$\Pi$ is CPA-secure if and only if $\MMM_{\Pi}$ is $(0,0,\FFF)$-$R$-computational-Narcissus-resilient.
\end{theorem}

\begin{proof}
Fix an adversary $\AAA=(\BBB,\CCC)$ as in Definition~\ref{def:indPK} and a security parameter $\lambda$. 
First note that
\begin{align*}
\underset{\substack{(S_0{=}0,S_1{=}1)\leftarrow\DDD\\b\leftarrow\{0,1\}\\\hat{b}\leftarrow\left(\AAA\leftrightarrows\MMM_{\Pi}(1^{\lambda},b)\right)}}{\Pr}[R(b,\hat{b})=1]
=
1-
\underset{\substack{(S_0{=}0,S_1{=}1)\leftarrow\DDD\\b\leftarrow\{0,1\}\\\hat{b}\leftarrow\left(\AAA\leftrightarrows\MMM_{\Pi}(1^{\lambda},b)\right)}}{\Pr}[R(b,\hat{b})=0]
=
1-
\underset{\substack{(S_0{=}0,S_1{=}1)\leftarrow\DDD\\b\leftarrow\{0,1\}\\\hat{b}\leftarrow\left(\AAA\leftrightarrows\MMM_{\Pi}(1^{\lambda},b)\right)}}{\Pr}[R(1-b,\hat{b})=1].
\end{align*}

Thus, $\MMM_{\Pi}$ is $(0,0,\FFF)$-$R$-computational-Narcissus-resilient if and  only if
\begin{align*}
\underset{\substack{(S_0{=}0,S_1{=}1)\leftarrow\DDD\\b\leftarrow\{0,1\}\\\hat{b}\leftarrow\left(\AAA\leftrightarrows\MMM_{\Pi}(1^{\lambda},b)\right)}}{\Pr}[R(b,\hat{b})=1]\leq\frac{1}{2}+\negl(\lambda),
\end{align*}
which concludes the proof as
\begin{align*}
\underset{\substack{(S_0{=}0,S_1{=}1)\leftarrow\DDD\\b\leftarrow\{0,1\}\\\hat{b}\leftarrow\left(\AAA\leftrightarrows\MMM_{\Pi}(1^{\lambda},b)\right)}}{\Pr}[R(b,\hat{b})=1]
=
\underset{\substack{
(k_d,k_e)\leftarrow\Gen(1^{\lambda})\\
b\leftarrow_R\{0,1\}\\
(m_0,m_1,{\sf state})\leftarrow \BBB(k_e) \\\
c\leftarrow\Enc_{k_e}(m_b)\\
\hat{b}\leftarrow\CCC({\sf state},c)
}}{\Pr}[\hat{b}=b].
\end{align*}
\end{proof}

\section{On the composability of Narcissus resiliency}

In general, Narcissus-resiliency is not preserved under composition, as illustrated in the following example.

\begin{example}
	Let $\lambda \in \bbN$, let $\cD$ be the uniform distribution over $\zo^\lambda \times \zo^\lambda$, and let $R((x,y), z) = 1 \iff z = x$. Consider the following two mechanisms: $\MMM_1((x,y)) = y$ and $\MMM_2((x,y)) = x \xor y$ (bit-wise xor). While both $\MMM_1, \MMM_2$ are $(0,0,\cF = \set{\cD})$-$R$-Narcissus-resilient (as each output alone does not reveal any information about $x$), the composed mechanism $\MMM=(\MMM_1,\MMM_2)$ is clearly not Narcissus-resilient: Let $\AAA$  be the algorithm that given $(u,v) \in \zo^\lambda \times \zo^\lambda$ as input, outputs $u \xor v$. Then it holds that 
	\begin{align*}
		\underset{\substack{(x,y)\leftarrow\DDD\\(y, x \xor y)\leftarrow\MMM(x,y)\\z\leftarrow\AAA(y, x \xor y)}}{\Pr}[R((x,y),z)=1] = 1 \quad \text{and} \quad \underset{\substack{(x,y)\leftarrow\DDD\\(x',y')\leftarrow\DDD\\(y, x \xor y)\leftarrow\MMM(x,y)\\z\leftarrow\AAA(y, x \xor y)}}{\Pr}[R((x',y'),z)=1] \leq 2^{-\lambda}.
	\end{align*}
\end{example}

Yet, in some settings Narcissus-resiliency is preserved under composition. First, as we showed, there are settings under which Narcissus-resiliency is equivalent to other security definitions which self-compose, such as differential privacy or security of one-way functions. In particular, in the equivalence we showed to differential privacy, the family $\cF$ consisted of all constant distributions over pairs of datasets (as in Definition~\ref{def:DP-pairs}), and the relation $R$ was ``everything'', in that the condition of Narcissus-resiliency needed to hold for all $R$. We are able to prove that a restricted type of composition still holds even without the ``for all'' quantifier over $R$ (for the case $\eps = 0$). Specifically,

\begin{theorem}[Composition of \cref{def:narcDP}]
	Let $\FFF$ be the family of (constant) distributions that for every pair of neighboring datasets $S,T\in\XXX^n$ contains the distribution $\DDD_{S,T}$, and let $R \colon  \XXX^n \times \cZ \rightarrow \zo$ be a relation.
	Let $\MMM_1, \ldots, \MMM_k$ be mechanisms such that each $\MMM_i \colon \XXX^n \rightarrow Y_i$ is $(0,\delta,\FFF)$-$R$-Narcissus-resilient. Then the mechanism $\MMM = (\MMM_1,\ldots,\MMM_k)\colon \XXX^n\rightarrow Y_1 \times \ldots \times Y_k$ is $(0,2\sum_{i=1}^k \delta_i,\FFF)$-$R$-Narcissus-resilient.
\end{theorem}
\begin{proof}
	Fix a pair of datasets $S_0,S_1$. Our goal is to show that for every algorithm $\AAA$:
	\begin{align}\label{eq:composition:goal}
		\underset{\substack{(S_0,S_1)\leftarrow\DDD\\b\leftarrow\{0,1\}\\(y_1,\ldots,y_k)\leftarrow\MMM(S_b)\\z\leftarrow\AAA(y_1,\ldots,y_k)}}{\Pr}[R(S_b,z)=1]&
		\leq
		\underset{\substack{(S_0,S_1)\leftarrow\DDD\\b\leftarrow\{0,1\}\\(y_1,\ldots,y_k)\leftarrow\MMM(S_{1-b})\\z\leftarrow\AAA(y_1,\ldots,y_k)}}{\Pr}[R(S_b,z)=1]+2 \sum_{i=1}^k \delta_i.
	\end{align}
	In the following, for a relation $R'$ define  $Im(R') = \set{(R'(S_0,z), R'(S_1,z)) \colon z \in \cZ} \subseteq \zo^2$, and let $R$ the relation from the theorem statement.
	If $(0,1), (1,0) \notin Im(R)$ or $\size{Im(R)} = 1$ then \cref{eq:composition:goal} immediately follows since in this case
	\begin{align*} \underset{\substack{(S_0,S_1)\leftarrow\DDD\\b\leftarrow\{0,1\}\\(y_1,\ldots,y_k)\leftarrow\MMM(S_b)\\z\leftarrow\AAA(y_1,\ldots,y_k)}}{\Pr}[R(S_b,z)=1] = \underset{\substack{(S_0,S_1)\leftarrow\DDD\\b\leftarrow\{0,1\}\\(y_1,\ldots,y_k)\leftarrow\MMM(S_{1-b})\\z\leftarrow\AAA(y_1,\ldots,y_k)}}{\Pr}[R(S_b,z)=1].
	\end{align*}
	In the following we assume that the other case holds where at least one of the pairs $\set{(1,0),(0,0)}$, $\set{(1,0),(0,1)}$, $\set{(1,0),(1,1)}$, $\set{(0,1), (0,0)}$, $\set{(0,1), (1,1)}$ is contained in $Im(R)$.
	We prove \cref{eq:composition:goal} by showing that \Enote{define indistinguishable in preliminaries}
	\begin{align}\label{eq:composition:goal2}
		\forall i \in [k]: \quad \MMM_i(S_0)\text{ and }\MMM_i(S_1)\text{ are }(0,2\delta_i)\text{-indistinguishable}.
	\end{align}
	Given that \cref{eq:composition:goal2} holds we conclude by composition that $\MMM(S_0)$ and $\MMM(S_1)$ are $(0,2 \sum_i \delta_i)$-indistinguishable, which in particular implies \cref{eq:composition:goal}. 
	
	Assume towards a contradiction that there exists $i \in [k]$ such that $\MMM_i(S_0)$ and $\MMM_i(S_1)$ are not $(0,2\delta_i)$-indistinguishable. Namely, there exists an event $B \subseteq \cY_i$ such that
	\begin{align*}
		\Pr[\MMM_i(S_0)\in B]>\Pr[\MMM_i(S_1)\in B]+2\delta_i.
	\end{align*}
	Fix two distinct elements $z_1,z_2 \in\cZ$ and let $\cA$ be the algorithm that given $y \in \cY_i$ as input, outputs $z_1$ if $y \in B$ and $z_2$ otherwise.
	We next define five relations $R_1,\ldots,R_5 \colon \cX^n \times \cZ \rightarrow \zo$:
	\begin{enumerate}
		\item $R_1(S,z) = \begin{cases} 1 & S=S_0 \land z = z_1 \\ 0 & \text{o.w.}\end{cases}\:\:$ (note that $Im(R_1) = \set{(1,0),(0,0)}$).
		
		\item $R_2(S,z) = \begin{cases} 1 & (S = S_0 \land z = z_1) \lor (S = S_1 \land z = z_2)\\ 0 & \text{o.w.}\end{cases}\:\:$ (note that $Im(R_2) = \set{(1,0),(0,1)}$).
		
		\item $R_3(S,z) = \begin{cases} 0 & S = S_1 \land z = z_1 \\ 1 & \text{o.w.}\end{cases}\:\:$ (note that $Im(R_3) = \set{(1,0),(1,1)}$).
		
		\item $R_4(S,z) = \begin{cases} 1 & S=S_1 \land z = z_2 \\ 0 & \text{o.w.}\end{cases}\:\:$ (note that $ Im(R_4) = \set{(0,1),(0,0)}$).
		
		\item $R_5(S,z) = \begin{cases} 0 & S = S_0 \land z = z_2 \\ 1 & \text{o.w.}\end{cases}\:\:$ (note that $Im(R_5) = \set{(0,1),(1,1)}$).
	\end{enumerate}
		%
		%
		%
		%
	Let $p_0 = \Pr[\MMM_i(S_0)\in B]$ and $p_1 = \Pr[\MMM_i(S_1)\in B]$ (recall that $p_0 - p_1 > 2\delta_i$). Furthermore, for $j \in [5]$ let
	\begin{align*}
		q_j \eqdef \underset{\substack{(S_0,S_1)\leftarrow\DDD\\b\leftarrow\{0,1\}\\y\leftarrow\MMM(S_b)\\z \la \AAA(y)}}{\Pr}[R_j(S_{b},z)=1]\quad\text{ and }\quad q_j' \eqdef \underset{\substack{(S_0,S_1)\leftarrow\DDD\\b\leftarrow\{0,1\}\\y\leftarrow\MMM(S_b)\\z \la \AAA(y)}}{\Pr}[R_j(S_{1-b},z)=1].
	\end{align*}
	Simple calculations yield the following:
	\begin{enumerate}
		\item $q_1 = \frac12  p_0 + \frac12 \cdot 0$ and $q_1' = \frac12 \cdot 0 + \frac12  p_1$,
		\item $q_2 = \frac12 p_0 + \frac12 (1 - p_1)$ and $q_2' = \frac12 p_1 + \frac12 (1 - p_0)$,
		\item $q_3 = 1 - (\frac12\cdot 0 + \frac12p_1)$ and $q_3' = 1 - (\frac12p_0 + \frac12 \cdot 0)$,
		\item $q_4 = \frac12\cdot 0 + \frac12\cdot (1-p_1)$ and $q_4' = \frac12 (1-p_1) + \frac12 \cdot 0$,
		\item $q_5 = 1 - (\frac12(1-p_0) + \frac12\cdot 0)$ and $q_5' = 1 - (\frac12\cdot 0 + \frac12(1-p_1))$,
	\end{enumerate}
	where in all cases, we first split into the $1/2$-probability event $b=0$ and then into the $1/2$-probability event $b=1$ (similarly to the proof of \cref{thm:DPnarc}).  We obtain for all $j \in [5]$ that $q_j - q_j' = \frac12(p_0 - p_1) > \delta$.
	
	In the following, let $j \in [5]$ such that $Im(R_j) \subseteq Im(R)$ (note that at least one such $j$ exists).
	We next translate the attack w.r.t. $R_j$ into an attack w.r.t. $R$.
	The attacker $\AAA'$, given $y$ as input, computes $z = \AAA(y)$ and outputs an (arbitrary) element $z'$ such that $(R(S_0,z'), R(S_1,z')) = (R_j(S_0,z), R_j(S_1,z))$. By definition, 
	\begin{align*}
		q \eqdef \underset{\substack{(S_0,S_1)\leftarrow\DDD\\b\leftarrow\{0,1\}\\y\leftarrow\MMM_i(S_b)\\z\leftarrow\AAA'(y)}}{\Pr}[R(S_{b},z)=1] = q_j\quad\text{ and }\quad q' \eqdef  \underset{\substack{(S_0,S_1)\leftarrow\DDD\\b\leftarrow\{0,1\}\\y\leftarrow\MMM_i(S_b)\\z\leftarrow\AAA'(y)}}{\Pr}[R(S_{1-b},z)=1] = q'_j.
	\end{align*}
	Since $q_j - q_j' > \delta$ we conclude that $q - q' > \delta$, yielding that $\MMM_i$ is not $(0,\delta,\FFF)$-$R$-Narcissus-resilient, a contradiction.  This completes the proof of \cref{eq:composition:goal2} and the proof of the theorem.
\end{proof}

	\section{Recognizing Extraction}\label{sec:RecognizingRec}

Following the discussion in \cref{sec:introIdentify}, we use the terminology of ``extraction'' rather than ``reconstruction'' in the setting where the dataset $S$ and the model $y$ are fixed. We formally define how to recognize a valid extraction attack on a given (fixed) model $y$ (as described in \cref{sec:introIdentify}), use the definition to explain the validity of some real-world attacks, and prove that in general the problem of verifying a valid attack is computationally hard assuming that cryptography exists.

\subsection{Defining Extraction via Kolmogorov Complexity}\label{sec:RecognizingRec:OurApproach}

In this paper, we think of extraction as a randomized process that can fail with some probability, so we use the following variant of \cref{def:intro:KLcomplexity}.

\begin{definition}[$K_{\cL}$-Complexity, a probabilistic version of \cref{def:intro:KLcomplexity}]\label{def:KLcomplexity}
	Let $\cL$ be a programming language.
	The $K_{\cL}$-complexity of a string $x$, denoted by $K_{\cL}(x)$, is the length of the shortest $\cL$-program that with probability $2/3$ outputs the string $x$ and halts.
\end{definition}

In theoretical results, this complexity is usually defined w.r.t. a fixed \emph{Universal Turing Machine} $\cU$ that gets a description $\cM$ of a Turing machine and outputs $\cU(\cM)$. Then, $K(x) = K_{\cU}(x)$ is defined by the shortest $\cM$ such that $\pr{\cU(\cM) = x} \geq 2/3$. This is good enough for \emph{asymptotic} analysis since for any programming language $\cL$ and every $x$ it holds that $K(x) \leq K_{\cL}(x) + O(1)$, where $O(1)$ denotes a constant that is independent of $\size{x}$.\footnote{Let $\cC$ be a $\cU$-program (i.e., a Turing machine description) that serves as a compiler from a language $\cL$ to $\cU$, i.e., give a $\cL$-program $\cM$, $C(\cM)$ outputs a $\cU$-program $\cM'$ that acts as $\cM$. So for every $x$, if it can be described using an $\cL$-program $\cM$, then it can be described using a $\cU$-program $\cC(\cM)$ of length $\size{\cM} + \size{\cC} = \size{\cM} + O(1)$ (note that $\size{\cC}$ is independent of $\size{x}$).} However, when we would like to deal with concrete quantities (where constants matter), we must explicitly specify a concrete programming language $\cL$, and therefore we do not omit it from the notation.

After fixing the programming language $\cL$, the quantity $K_{\cL}(x)$ can be thought of as the analogous to the ``entropy" of the string $x$. As with entropy, we can also consider a conditional version of the $K$-complexity (the analog of conditional entropy) as well as ``mutual $K$-information" (the analog of mutual information) \cite{ZvonkinLevin70,Levin73,Trakhtenbrot84,LMocas93,Grunwald2003-GRNKCA,LiuP22,BallLMP23}:

\begin{definition}[Conditional $K_{\cL}$-Complexity, and mutual $K_{\cL}$-information]\label{def:KI}
	Given a programming language $\cL$ and two strings $x,y$, we define the \emph{conditional $K_{\cL}$-complexity of $x$ given $y$}, denoted by $K_{\cL}(x \mid y)$, as the length of the shortest program in $\cL$ that given $y$ as input, with probability $2/3$ outputs $x$ and halts. We define the \emph{mutual $K_{\cL}$-information} of $x,y$ as $KI_{\cL}(x ; y) = K_{\cL}(x) - K_{\cL}(x \mid y)$. 
\end{definition}

Intuitively, $KI_{\cL}(x ; y)$ is high if it is possible to extract $x$ from $y$ using a short $\cL$-program, but it is not possible to do so without $y$.
In our context of extracting an input example $x$ from a model $y$, we should also allow to output elements $z$ that are ``similar" to $x$ (which are considered as a valid extraction of $x$ under some metric), or perhaps the goal is to output a list of elements that some of them are close to training examples.
Therefore, we would be interested in an extension of the $K$-complexity into a set of strings which allows much more flexibility.

\begin{definition}[Extension of \cref{def:KLcomplexity,def:KI} for sets]\label{def:KLcomplexitySet}
	Let $\cL$ be a programming language.
	The $K_{\cL}$-complexity of a set of strings $X$, denoted by $K_{\cL}(X)$, is the length of the shortest  $\cL$-program that with probability $2/3$, outputs a string $x \in X$ and halts. 
	Similarly to \cref{def:KI}, we define $K_{\cL}(X \mid y)$ as the length of the shortest $\cL$-program that given an input $y$, outputs w.p. $2/3$ a string $x \in X$ and halts, and we define $KI_{\cL}(X; y) = K_{\cL}(X) - K_{\cL}(X\mid y)$.
\end{definition} 

Note that by definition, $KI_{\cL}(X; y) = \min_{x \in X} K_{\cL}(x) - \min_{x \in X} K_{\cL}(x \mid y)$, and these minimums are not necessarily realized by the same $x$.

A few examples are in order. In all the following examples, we use $\cL$ as a fixed Universal Turing Machine language, and for $x = (x^1,\ldots,x^d), z = (z^1,\ldots,z^d) \in \zo^d$ we denote by $\NHdist(x, z) = \frac1{d}\cdot \size{\set{i \in [n] \colon x^i \neq z^i}}$ the \emph{Normalized Hamming Distance} between $x$ and $z$.

\begin{example}\label{example:K:simple_x}
For $x \in \zo^d$ define \haim{This $R$ can be confused with  the relation from the definition of self resilience } $R_x = \set{(z_1,z_2) \in \zo^{2d} \colon \min_{i \in \set{1,2}}{\NHdist(x,z_i)} \leq 1/2},$ (i.e., in this example, $R_x$ defines the task of outputting two strings that at least one of them agrees with $x$ on at least half of the bits).
Note that $K_{\cL}(R_x) = O(\log d)$ for any $x$, because the program that outputs a uniformly random $(z_1,z_2) \la \zo^{2 d}$ -- 
 e.g., a for loop executing for $2d$ iterations and outputing a random bit in each -- 
 satisfies $(z_1,z_2) \in R_x$ with probability $3/4$.
\end{example}

\begin{example}\label{example:K:complex_x}
	For $x \in \zo^d$ define $R_x$ as the set of all the tuples of the form $z = (z_1,\ldots, z_{m})$, for $m \leq \poly(d)$ (for some fixed polynomial), such that $\min_{i \in [m]}{\NHdist(x,z_i)} \leq 1/4$ (i.e., at least one of the $z_i$'s agrees with $3/4$ of the bits of $x$). For a uniformly random $x \la \zo^d$, it holds with high probability that $K_{\cL}(R_x) \geq (1 - o(1)) d/2$ (i.e., for most strings $x$, any algorithm that aims to output a string that agrees with $3/4$ of the bits of $x$, essentially must memorize $\approx d/2$ bits of information about $x$).  Indeed, given $y = (x^1,\ldots,x^{d/2})$ (the first half of the bits of $x$) as input, we have $K(R_x \mid y) = O(\log d)$ because the algorithm that chooses $m = 2$ uniformly random strings $(z_1, z_2) \la \zo^{d/2}$ and outputs $z = (y \circ z_1,\:  y \circ z_2)$ satisfies $z \in R_x$ with probability $3/4$.
\end{example}

\begin{example}\label{example:K:S}
	Let $S = (x_1,\ldots,x_n) \in (\zo^d)^n$, and define $R_S$ as the set of all the tuples $z = (z_1,\ldots, z_{m})$ such that there exist $i \in [n]$ and $j \in [m]$ with $\NHdist(x_i,z_j) \leq 1/4$. Assume $n,m \leq \poly(d)$ (for some fixed polynomial). Then as in \cref{example:K:complex_x}, it can be shown that when $S$ consists of uniformly random strings, then $K_{\cL}(R_S) \geq (1 - o(1)) d/2$ with high probability, meaning that solving this problem requires to memorize at least $\approx d/2$ bit of information from one of the strings.
\end{example}

We are now ready to reformulate the extraction definitions of \citet{CarliniLargeModels21,CarliniDiffusion23}.

\begin{definition}[Extractable Information]\label{def:extraction}
	Let $S$ be a set of elements over a domain $\cX$, let $R:\XXX^*\times\{0,1\}^*\rightarrow\{0,1\}$ be an extraction relation (i.e., $R(S,z)=1$ means that $z$ is considered a valid extraction of elements in $S$), let $\cL$ be a programming language, and let $R_S = \set{z \colon R(S,z) = 1}$.
	We say that a model $y$ contains $(R,\cL,\tau)$-extractable information about $S$ if
	\begin{align*}
		KI_{\cL}(R_S ; y) \geq \tau.
	\end{align*} 
\end{definition}

In order to prove that $S$ is extractable from a model $y$, it suffices to present an extractor:

\begin{definition}[Extractor, redefinition of \cref{def:intro:extraction}]\label{def:extractor}
	Let $R \colon \cX^* \times\{0,1\}^* \rightarrow \zo$ be an extraction relation and denote $R_S = \set{z \colon R(S,z) = 1}$. Let $q \in [0,1)$ be a quality parameter. We say that an $\cL$-program $\cA$ is an $(\cL,R,q)$-extractor 
	of $S \in \cX^*$ from $y$  if the following holds:
	\begin{enumerate}
		\item $\pr{\cA(y) \in R_S} \geq 2/3$, and\label{item:Aextracts}
		\item $K_{\cL}(R_S) \geq \frac{\size{\cA}}{1-q}$.\label{item:HardExtraction}
	\end{enumerate}
\end{definition}

Note that when $q$ is closer to $1$ is means a more significant extraction.

\begin{claim}
	If there exists an $(\cL,R,q)$-extractor $\cA$ of $S$ from $y$, then $y$ contains $(\cL,R,\tau = \frac{q}{1-q}\cdot \size{\cA})$-extractable information about $S$ (\cref{def:extraction}).
\end{claim}
\begin{proof}
	Compute
	\begin{align*}
		KI_{\cL}(R_S ; y) = K_{\cL}(R_S) - K_{\cL}(R_S \mid y) \geq \frac{\size{\cA}}{1-q} -  \size{\cA} = \frac{q}{1-q}\cdot \size{\cA}.
	\end{align*}
\end{proof}

\subsection{Capturing Interesting Extractions using the Predicate}

In the previous examples (e.g., \cref{example:K:S}), we considered simple relations $R$ that only check for a small distance from an example point. But in the real-world, when thinking about $S$ as the training dataset and on $y$ as the trained model (over $S$), it is very likely that $S$ will contain duplicated data, and in such cases, it might be considered ok to memorize such data. As a concrete example, if we think about a chat-bot model, then we can query it with "Print the Bible" and it will likely do that. I.e., even though it is indeed a valid extraction (as the $K$-complexity of the Bible is large), this is not an interesting one since the Bible content is probably duplicated in many training examples. In order to define more interesting relations $R$, it is reasonable to consider extraction of only \emph{$k$-Eidetic Memorizated} strings $x$ (\cite{CarliniLargeModels21}) for some small threshold $k$, which means that such $x$'s  only appear at most $k$ times in $S$.

\subsection{Real-World Examples}\label{sec:real-world-attacks}

In order to illustrate the generality of our definitions, we next pick three different types of real-world attacks, and explain what is happening in each of them in terms of \cref{def:extractor}.

\paragraph{Extraction from Large Language Models \cite{CarliniLargeModels21}.}

\citet{CarliniLargeModels21} present extraction attacks on GPT-2 by essentially generating many samples from GPT-2 when the model is conditioned on (potentially empty) prefixes. Then they sort each generation according to some metric to remove duplications, and this gives a set of potentially memorized  texts from training examples. In one of their methods, they choose the top-$m$ interesting ones by assigning a score to each such text which combines the compressibility of the text using some compression algorithm like zlib \cite{zlib}, and the likelihood of it, e.g., using a ``perplexity" measure \cite{Carlini0EKS19} (higher compressibility and likelihood increase the score). They show that noticeable fraction of such strings are "Eidetic Memorizated", meaning that they only appear a small number of times in the training data (which makes them more interesting).

In terms of our extraction methodology, the training dataset $S = (x_1,\ldots,x_n)$ consists of webpages, and an \emph{interesting extractable text} is a string $s$ that: 
(1) has high entropy (i.e., high $K$-complexity, which can be estimated using heuristics like zlib compression) and (2) appears only  small number of times in the training examples (here, ``appear" in webpage $x_i$ means $s \subseteq x_i$).
The attack outputs a list of strings $z = (z_1,\ldots,z_m)$, and succeed when there exists a sub-list $(z_{j_1},\ldots,z_{j_{\ell}})$ such that each string $z_{j_i}$ is interesting and extractable, and the $K$-complexity of the entire list is higher than some threshold (which, again, can be estimated using heuristics like zlib compression).

When the $K$-complexity of the resulting sub-list is much higher than the length of the attacking code that extracted this information from the GPT-2 model, the attack is considered to have high quality.

\paragraph{Extraction from Diffusion Models \cite{CarliniDiffusion23}.}

\citet{CarliniDiffusion23} present several extraction attacks on Diffusion models, where all the attacks assume that the training data is given in advance. As one concrete example, they managed to extract training images from Imagen (\cite{SahariaCSLWDGLA22}), a $2$ billion parameter text-to-image diffusion model, by first searching for the captions of "outliers" in the training data (that is, training images that are less similar to other training images according to some metric), and then showed that when querying the model on such captions, some of them result with a good approximation of unique training images.

In terms of our extraction methodology, the end program that queries the model on a specific caption that leads to extraction of a unique training image is extremely short compared to the $K$-complexity of the image that was extracted. So this should be considered a very high-quality extraction to a predicate of the following type (for some threshold parameters $T > t$):

\begin{align*}
	R(S=(x_1,\ldots,x_n),z) = 1 \iff \exists i \in [n]\text{ s.t. }\paren{\ell_2(x_i,z) < t \text{ and }\forall j\in [n]\setminus \set{i}: \ell_2(x_i,x_j) > T}. 
\end{align*}

\paragraph{Extraction from Binary Classifiers trained by Large Neural Networks \cite{HaimVYSI22}.}

\citet{HaimVYSI22} present a general attack that works on large class of binary classifiers trained by Neural networks with many parameters $p$. Roughly, their attack works as follows: Given the parameters of the model $y \in \bbR^p$, they define an optimization problem (that depends on $y$) on variables $z_1,\ldots,z_m \in \bbR^d$ and $\lambda_1,\ldots,\lambda_m \in \bbR$ and optimize it, and show that in the regime $p > n\cdot d$, the solution $z_1,\ldots,z_m$ is likely to contain a good approximation of many data points.

As one example of their experimental evaluation, they trained the model on the training images of CIFAR10 on vehicles vs. animals classification, and showed that their optimization program, which is defined based on the model $y \in \bbR^p$ given as input, managed to extract a good approximation of many training images (the top $45$ extractions are presented in their paper). 

In terms of our extraction methodology, given the training data $S = (x_1,\ldots,x_n)$, we can define the extraction relation
\begin{align*}
	R(S=(x_1,\ldots,x_n),z = (z_1,\ldots,z_m)) = 1 \iff \size{\set{i \in [n] \colon \exists j \in [m]\text{ s.t. }z_j\text{ is ``close" to }x_i}} \geq t,
\end{align*}
for some parameter $t$ (say, $t = 45$), where ``close" is according to some metric.
Under the reasonable assumption that $K(R_S)$ is much higher than the length of their attack (which is given the model $y$ as input), this is considered a high quality extraction.

\subsection{Inherent Limitations of \cref{def:extractor}}

The main limitation of our \cref{def:extractor} is that \cref{item:HardExtraction} (lower-bounding the Kolmogorov Complexity) is not verifiable, as computing the $K$-complexity is an intractable problem. While in practice the $K$-complexity can be estimated using Heuristics like zlib compression \cite{zlib}, any such Heuristic can fail to determine some highly compressible patterns, and as we demonstrate next in \cref{sec:hard-of-ver}, this problem is indeed inherert assuming that basic cryptography exists.

A well studied tractable variant of the $K$-complexity is the $t(\cdot)$-time bounded $K$-complexity, denoted by $K^t$-complexity, where $K^t(x)$ is the length of the shortest program that outputs $x$ within $t(\size{x})$ steps (\cite{Kolmogorov68,Sisper83,Trakhtenbrot84,Ko86}). However, there is no efficient algorithm for computing the $K^t$-complexity.\footnote{As surveyed by \cite{Trakhtenbrot84}, the problem of efficiently determining the $K^t$-complexity for $t(n) = \poly(n)$ predates the theory of NP-completeness and was studied in the Soviet Union since the 60s as a candidate for a problem that requires “brute-force search” (see Task 5 on page 392 in \cite{Trakhtenbrot84}).} Very recently, in a sequence of breakthrough results in cryptography and complexity theory \cite{LiuP20,LiuP21,LiuP22a,LiuP22}, it was shown that (a close variant of) this problem is hard if and only if one-way functions, the most basic cryptographic primitives, exist. Since a real-world verifier is a polynomial-time algorithm, this means that switching from $K(\cdot)$ to the $K^{\poly(n)}(\cdot)$ variant will essentially not make any difference. Actually, the hardness of verification is inherent, and in the next two sections we show, under common cryptographic assumptions, how to construct triplets $(\cA,x,y)$, with efficient $\cA$, such that it is computationally hard to verify whether $\cA$ extracts information about $x$ from $y$.

\remove{

In order to prove that hardness of verification is inherent, we first start by defining what is hardness of verification.
Throughout the rest of this section, we use $\cL$ as a fixed representation of a Universal Turing machine, and assume that all algorithms are Turing machines. We also use $R$ as the identity function. 

\Enote{Maybe remove this definition}
\begin{definition}\label{def:hard-to-verify}
	Let $\Attacker$ be a \ppt algorithm that outputs $(A,x,y)$ and let $f_{\tau}(A,x,y)$ be the predicate that is $1$ iff $A$ extracts $\tau$-information about $x$ from $y$ (\cref{def:extractor}). We say that $\Attacker$ provides $\tau$-extraction example that is $\beta$-hard to verify if for any \ppt algorithm $V$ (verifier) it holds that
	\begin{align}\label{eq:verifier}
		\ppr{(A,x,y) \la \Attacker, \:\: b \la V(A,x,y)}{b = f_{\tau}(A,x,y)} \geq 1 - \beta.
	\end{align}
	We say that it is $\beta$-black-box hard to verify if \cref{eq:verifier} only holds for verifiers who use only black-box access to $A$.
\end{definition}

In the following section, we show $\Attacker$ that is $\approx 1/2$-hard to verify, even when $\tau \gg K(x)$ with probability $1/2$ (i.e., while it looks like $A$ extract a lot of information about $x$ from $y$, actually in half of the cases we could get $x$ using a very short program without even the need of $y$).
}

\subsubsection{Hardness of Verification}\label{sec:hard-of-ver}

Our harness example is based on the existence of Pseudorandom Generator (PRG).

\begin{definition}[Pseudorandom Generator (PRG)]\label{def:PRG}
	A PRG is a \ppt algorithm $G$ that maps $d$ bits to $\ell(d)$ bits (for $\ell(d) > d$) such that
	for every \ppt algorithm $V$ and every $d \in \bbN$ it holds that
	\begin{align*}
		\pr{V(G(\cU_d)) = 1} \leq \pr{V(\cU_{d + \ell(d)}) = 1} + \negl(d),
	\end{align*}
	where $\cU_m$ denotes the uniform distribution over $\zo^m$.
\end{definition}

\begin{lemma}\label{lemma:bb-hard-to-ver}
	Assume there exists a PRG $G$ that maps $d$ bits to $d^4$ bits. Let $\cL$ be a Universal Turing Machine programming language with oracle access to $G$.
	Then there exists a \ppt algorithm $\Attacker$ such that on input $1^d$ and $b \in \zo$, outputs a deterministic program $\cA$ and two strings $x, y \in \zo^{d^4}$ such that the following holds w.r.t. a random execution $(\cA,x,y) \la \Attacker(1^d,b)$:
	\begin{enumerate}
		
		\item $\size{\cA} = d^2 + O(1)$ and $\cA(y) = x$ (which is an evidence for $K_{\cL}(x \mid y) \leq d^2 + O(1)$).\label{item:BB-hard:sizeA}
		
		\item If $b=0$ then $K_{\cL}(x) \leq d + O(1)$ (and in particular, there is a deterministic $\poly(d)$-time $\paren{d + O(1)}$-size algorithm that outputs $x$).\label{item:BB-hard:b0}
		
		\item If $b = 1$ then $x$ is chosen uniformly over $\zo^{d^4}$ (which yields $K_{\cL}(x) \geq (1 - o(1)) d^4$ w.h.p.),\label{item:BB-hard:b1}
		
		\item For every \ppt $\cV$ it holds that\label{item:BB-hard:V}
		\begin{align*}
			\ppr{b \la \zo, \: (\cA,x,y) \la \Attacker(1^d,b)}{\cV(\cA,x,y) = b} \leq 1/2+\negl(d). 
		\end{align*}
		
	\end{enumerate}
\end{lemma}

Note that when $b=1$, $\cA$ extracts most of the information about $x$ from $y$ (quality $1 - \frac{1 + o(1)}{d^2}$), while the case $b=0$ is clearly not an extraction as $K_{\cL}(x) \ll \size{\cA}$. Yet, by \cref{item:BB-hard:V}, any \ppt verifier $\cV$ who only sees $\cA,x,y$ cannot distinguish between the two cases.

\begin{proof}
The proof holds by considering the algorithm $\Attacker$ defined in \cref{alg:attacker}. 

\begin{algorithm}[$\Attacker$]\label{alg:attacker}
	\item Input: $1^d$ and $b \in \zo$.
	\item Oracle: PRG $G\colon \zo^d \rightarrow \zo^{d^4}$.
	\item Operation:~
	\begin{enumerate}
		\item If $b = 0$: Sample $s \la \zo^d$ and let $x = (x_1,\ldots,x_{d^4}) = G(s)$.
		\item Else (i.e., $b=1$): Sample $x = (x_1,\ldots,x_{d^4}) \la \zo^{d^4}$.
		\item Let $y = (x_{1},\ldots,x_{d^4 - d^2})$.\label{step:y}
		\item Let $\cA$ be the $(d^2 + O(1))$-size algorithm that has $x' = (x_{d^4 - d^2 + 1},\ldots,x_{d^4})$ hard-coded, and given $y$ as input, outputs $y \circ x'$ and halts.\label{step:A}
		\item Output $(\cA,x,y)$.
	\end{enumerate}
\end{algorithm}

Note that \cref{item:BB-hard:sizeA,item:BB-hard:b0,item:BB-hard:b1} of the lemma holds by construction. 
Furthermore, \cref{item:BB-hard:V} immediately holds since $G$ is a PRG.

\end{proof}

In the above example, we defined a hardness problem w.r.t. the predicate $R(x,z) = 1 \iff z = x$. But it can be easily be relaxed to a more realistic predicate of the form $$R(S = (x_1,\ldots,x_n), z = (z_1,\ldots,z_{m})) = 1 \iff \exists i \in [n], j \in [m]\text{ s.t. } \set{\NHdist(x_i, z_j)} \leq 1/2 - \alpha,$$ for any fixing of a constant $\alpha > 0$ and $n,m \leq \poly(d)$ (i.e., when the task is to output a list of strings that at least one of them has non-trivial agreement with one of the elements of $S$).\footnote{The idea is to replace \stepref{step:y} of \cref{alg:attacker} with $y = (x_{1},\ldots,x_{\alpha d^4 - d^2 \log d})$, where $x$ is one of the elements in $\cS$, and in \stepref{step:A} to define $\cA$ as the $O(d^2 \log d)$-size program that has $x' = (x_{\alpha d^4 - d^2 \log d}, \ldots, x_{\alpha d^4 + d^2 \log d})$ hard-coded, and given $y$ as input, samples $x'' \la \zo^{(1-\alpha)d^4 - d^2 \log d}$ and outputs $z = y \circ x' \circ x''$. By concentration bounds of the binomial distribution, it holds that $\pr{\NHdist(x, z) \leq 1/2 - \alpha} \geq 1- \negl(d)$. But without memorizing $x'$ or at least $\Omega(d^2 \log d)$ bits of information about one $x \in \cS$, the probability that a random guess of the bits after $y$ would lead to $1/2 - \alpha$ Normalized Hamming distance from an element $x \in \cS$ is $\negl(d)$.}

%

        \section{Alternative proof that differential privacy can be expressed via Narcissus resiliency}\label{appendix:dp}

\begin{theorem} Let $\MMM:\XXX^n\rightarrow Y$ be an algorithm and let $\FFF$ be the family of distributions over $\XXX^n$ that for every neighboring datasets $S,T\in\XXX^n$ contains the  uniform distribution over $\{S,T\}$.
Then,
\begin{enumerate}
    \item[(1)] If $\MMM$ is $(\eps,\delta)$-differentially private then $\MMM$ is $(\eps,\delta,\FFF)$-$R$-narcissus-resilient for all $R$.
    \item[(2)]  If $\MMM$ is not $(\eps,\delta)$-differentially private then there exists $R$ such that $\MMM$ is not $(\tilde{\eps},\tilde{\delta},\FFF)$-$R$-narcissus-resilient for $R$.
\end{enumerate}
\end{theorem}

\begin{proof}
Suppose that $\MMM$ is $(\eps,\delta)$-DP. Let $\DDD$ be a distribution over neighboring datasets $\{S_0 , S_1\}$, let $R$ be a relation, and let $\AAA$ be an attacker. We calculate, 
\begin{align*}
\underset{\substack{S\leftarrow\DDD\\y\leftarrow\MMM(S)\\z\leftarrow\AAA(y)}}{\Pr}[R(S,z)=1]&
\leq
e^{\eps}\cdot
\underset{\substack{S\leftarrow\DDD\\T\leftarrow\DDD\\y\leftarrow\MMM(T)\\z\leftarrow\AAA(y)}}{\Pr}[R(S,z)=1]+\delta
=
e^{\eps}\cdot
\underset{\substack{S\leftarrow\DDD\\T\leftarrow\DDD\\y\leftarrow\MMM(S)\\z\leftarrow\AAA(y)}}{\Pr}[R(T,z)=1]+\delta.
\end{align*}

This concludes the analysis for the first direction. Now suppose that $\MMM$ is not $(\eps,\delta)$-DP. There exist neighboring datasets $S_0,S_1$ and event $B$ such 
that 
$$
\Pr[\MMM(S_0)\in B]>e^{\eps}\cdot\Pr[\MMM(S_1)\in B]+\delta.
$$

We need to show that there is a relation $R$, a distribution in $\FFF$, and an adversary $\AAA$ for which the condition of Narcissus-Resiliency does not hold. We let $\AAA$ be the identify function and let $\DDD$ be the uniform distribution over $\{S_0,S_1\}$. Now let $R$ be such that $R(S_b,y)=1$ if and only if $S_b=S_0$ and $y\in B$. With these notations, to contradict Narcissus-Resiliency, we need to show that 
$$\underset{\substack{S_b\leftarrow\DDD\\y\leftarrow\MMM(S_b)\\y\leftarrow\AAA(y)}}{\Pr}[R(S_b,y)=1]\gg\underset{\substack{S_b\leftarrow\DDD\\S_c\leftarrow\DDD\\y\leftarrow\MMM(S_b)\\y\leftarrow\AAA(y)}}{\Pr}[R(S_c,y)=1].$$
We now calculate each of these two probabilities:
\begin{align*}
\underset{\substack{S_b\leftarrow\DDD\\y\leftarrow\MMM(S_b)\\y\leftarrow\AAA(y)}}{\Pr}[R(S_b,y)=1] & = \frac{1}{2}\cdot\left(\underset{\substack{y\leftarrow\MMM(S_0)\\y\leftarrow\AAA(y)}}{\Pr}[R(S_0,y)=1]+\underset{\substack{y\leftarrow\MMM(S_1)\\y\leftarrow\AAA(y)}}{\Pr}[R(S_1,y)=1]\right)\\
& = \frac{1}{2}\cdot\left(
\Pr[\MMM(S_0)\in B]
+0\right)\\
&=\frac{1}{2}\cdot\Pr[\MMM(S_0)\in B].\\
\end{align*}

On the other hand,

\begin{align*}
&\underset{\substack{S_b\leftarrow\DDD\\S_c\leftarrow\DDD\\y\leftarrow\MMM(S_b)\\y\leftarrow\AAA(y)}}{\Pr}[R(S_c,y)=1]=\\
&=\Pr[b=c=0]\cdot\underset{\substack{y\leftarrow\MMM(S_0)\\y\leftarrow\AAA(y)}}{\Pr}[R(S_0,y)=1]
+
\Pr[b=c=1]\cdot\underset{\substack{y\leftarrow\MMM(S_1)\\y\leftarrow\AAA(y)}}{\Pr}[R(S_1,y)=1]\\
&\qquad\qquad+
\Pr[b>c]
\cdot\underset{\substack{y\leftarrow\MMM(S_1)\\y\leftarrow\AAA(y)}}{\Pr}[R(S_0,y)=1]
+
\Pr[b<c]
\cdot\underset{\substack{y\leftarrow\MMM(S_0)\\y\leftarrow\AAA(y)}}{\Pr}[R(S_1,y)=1]\\
&=\frac{1}{4}\cdot 
\Pr[\MMM(S_0)\in B]
+0+\frac{1}{4}\cdot \Pr[\MMM(S_1)\in B]+0\\
&<\frac{1}{4}\cdot 
\Pr[\MMM(S_0)\in B]+\frac{1}{4}\left(e^{-\eps}\cdot \Pr[\MMM(S_0)\in B]-\delta\cdot e^{-\eps}\right)\\
%
%
&=\frac{1+e^{-\eps}}{4} \cdot \Pr[\MMM(S_0)\in B]-\frac{\delta\cdot e^{-\eps}}{4}\\
&=\frac{1+e^{-\eps}}{2} \cdot \underset{\substack{S_b\leftarrow\DDD\\y\leftarrow\MMM(S_b)\\y\leftarrow\AAA(y)}}{\Pr}[R(S_b,y)=1]-\frac{\delta\cdot e^{-\eps}}{4}\\
&=e^{-\tilde{\eps}} \cdot \underset{\substack{S_b\leftarrow\DDD\\y\leftarrow\MMM(S_b)\\y\leftarrow\AAA(y)}}{\Pr}[R(S_b,y)=1]-\tilde{\delta}\cdot e^{-\tilde{\eps}}~\mbox{where}~\tilde{\eps}=-\ln((1+e^{-\eps})/2)~\mbox{and}~\tilde{\delta}=\delta\cdot e^{-\eps}/4.
\end{align*}
\end{proof}

\end{document}